\documentclass[12pt,a4paper]{article}

\usepackage{amsmath, amsthm, amssymb,accents,mathrsfs,mathtools}
\usepackage{authblk}
\usepackage{enumerate}
\usepackage{pdflscape}

\usepackage[pdftex,plainpages=false,hypertexnames=false,pdfpagelabels]{hyperref}
\usepackage{xcolor}
\definecolor{dark-red}{rgb}{0.7,0.25,0.25}
\definecolor{dark-blue}{rgb}{0.15,0.15,0.55}
\definecolor{medium-blue}{rgb}{0,0,.8}
\definecolor{DarkGreen}{RGB}{0,150,0}
\definecolor{rho}{named}{red}
\hypersetup{
   colorlinks, linkcolor={purple},
   citecolor={medium-blue}, urlcolor={medium-blue}
}

\usepackage{fullpage}
\theoremstyle{plain}
\newtheorem{thm}{Theorem}[section]
\newtheorem{cor}[thm]{Corollary}

\newtheorem{lem}[thm]{Lemma}
\newtheorem{prop}[thm]{Proposition}
\newtheorem*{mainres}{Main Result}
\newtheorem*{thm*}{Theorem}
\theoremstyle{definition}
\newtheorem{defn}[thm]{Definition}
\theoremstyle{remark}
\newtheorem{ex}[thm]{Example}
\newtheorem{rem}[thm]{Remark}

\numberwithin{equation}{section}

\DeclareMathOperator{\Diff}{Diff}

\DeclareMathOperator{\End}{End}

\DeclareMathOperator{\Lie}{Lie}

\DeclareMathOperator{\Mob}{M\ddot{o}b}
\DeclareMathOperator{\Rot}{Rot}
\DeclareMathOperator{\spann}{span}
\DeclareMathOperator{\supp}{supp}

\newcommand{\abs}[1]{\left| #1 \right|}
\newcommand{\ip}[1]{\langle #1 \rangle}
\newcommand{\bil}[1]{\left( #1 \right)}
\newcommand{\norm}[1]{\left\| #1 \right\|}

\newcommand{\opp}{\mathsf{o}}

\def\semicolon{;}
\def\applytolist#1{
    \expandafter\def\csname multi#1\endcsname##1{
        \def\multiack{##1}\ifx\multiack\semicolon
            \def\next{\relax}
        \else
            \csname #1\endcsname{##1}
            \def\next{\csname multi#1\endcsname}
        \fi
        \next}
    \csname multi#1\endcsname}

\def\calc#1{\expandafter\def\csname c#1\endcsname{{\mathcal #1}}}
\applytolist{calc}QWERTYUIOPLKJHGFDSAZXCVBNM;
\def\bbc#1{\expandafter\def\csname bb#1\endcsname{{\mathbb #1}}}
\applytolist{bbc}QWERTYUIOPLKJHGFDSAZXCVBNM;
\def\bfc#1{\expandafter\def\csname bf#1\endcsname{{\mathbf #1}}}
\applytolist{bfc}QWERTYUIOPLKJHGFDSAZXCVBNM;
\def\sfc#1{\expandafter\def\csname s#1\endcsname{{\sf #1}}}
\applytolist{sfc}QWERTYUIOPLKJHGFDSAZXCVBNM;

\def\rmi{{\mathrm i}}
\def\rme{{\mathrm e}}

\newcommand{\noshow}[1]{}

\usepackage{wasysym}

\newcounter{wightman}

\title{Non-unitary Wightman CFTs and non-unitary vertex algebras}

\author[1]{Sebastiano Carpi\thanks{{\tt carpi@mat.uniroma2.it}}}
\author[2]{Christopher Raymond\thanks{{\tt christopher.raymond@uni-hamburg.de}}}
\author[1]{Yoh Tanimoto\thanks{{\tt hoyt@mat.uniroma2.it}}}
\author[3]{James E.\! Tener\thanks{{\tt james.tener@anu.edu.au}}}

\affil[1]{Dipartimento di Matematica, Universit\`a di Roma Tor Vergata,\authorcr
   Via della Ricerca Scientifica 1, I-00133 Roma, Italy}
\affil[2]{Department of Mathematics, University of Hamburg, 20148 Hamburg, Germany}
\affil[3]{Mathematical Sciences Institute, Australian National University,\authorcr
Canberra, ACT 2600, Australia}
\date{}

\begin{document}

\maketitle

\begin{abstract}
We give an equivalence of categories between: (i) M\"obius vertex algebras which are equipped with a choice of generating family of quasiprimary vectors, and (ii) (not-necessarily-unitary) M\"obius-covariant Wightman conformal field theories on the unit circle.
We do not impose any technical restrictions on the theories considered (such as finite-dimensional conformal weight spaces or simplicity), yielding the most general equivalence between these two axiomatizations of two-dimensional chiral conformal field theory.
This provides new opportunities to study non-unitary vertex algebras using the lens of algebraic conformal field theory and operator algebras, which we demonstrate by establishing a non-unitary version of the Reeh-Schlieder theorem. 
\end{abstract}

{\hypersetup{linkcolor=black}
\setcounter{tocdepth}{2}
\tableofcontents
}

\section{Introduction}

It is a fundamental mathematical challenge to establish a rigorous axiomatization of quantum field theory (QFT), and in general this problem remains wide open except in very specialized contexts.
In recent years, axiomatic QFT has received particular attention in the context of two-dimensional chiral conformal field theories (CFTs), as these theories are sufficiently structured to enable a rigorous mathematical treatment while at the same time exhibiting a wide variety of mathematical connections (to operator algebras and subfactors, to representation theory and modular tensor categories, to vector-valued modular forms, and to many other areas).
There are many proposed axiomatizations of two-dimensional chiral CFTs, each of which captures different aspects of the physical theory, and none of which have been rigorously demonstrated to provide a complete description of chiral CFT.
It is conjectured that these different axiomatizations are essentially equivalent, and there have been recent breakthroughs in comparing different axiomatizations under certain technical hypotheses \cite{CKLW18,GRACFT1}.

In this article we demonstrate the equivalence of two well-known axiomatizations of two-dimensional chiral CFTs.
We establish this equivalence in the most general way possible, without any reliance on auxiliary technical hypotheses or restrictions on the models under consideration, such as the existence of an invariant inner product (a positive-definite sesquilinear form giving rise to ``unitarity'').
Proving equivalences at this level of generality has largely been viewed as an aspirational (but not necessarily feasible) goal of axiomatic QFT, which we achieve here through a detailed analysis of the mathematical structures in question.

The first axiomatization that we consider is the non-unitary version of the (bosonic) Wightman axioms on the unit circle $S^1 \subset \bbC$, with M\"obius symmetry (i.e.\! symmetry group $\Mob:=\mathrm{PSU}(1,1)$, the holomorphic automorphisms of the unit disk).
The key data of such a theory is a collection $\cF$ of operator-valued distributions (or Wightman fields) acting on a common invariant vector space of states $\cD$, along with a compatible positive-energy representation of $\Mob$.

The second axiomatization that we consider is ($\bbN$-graded, bosonic) M\"obius vertex algebras.
These are vertex algebras graded by non-negative integer conformal dimensions, with symmetry given by the complexified Lie algebra $\mathfrak{su}(1,1)_\bbC \cong \mathfrak{sl}_2(\bbC)$.

We prove the following main result.

\begin{mainres}
There is a natural equivalence of categories between non-unitary M\"obius-covariant Wightman conformal field theories on $S^1$ and M\"obius vertex algebras equipped with a family of quasiprimary generators.
\end{mainres}

Our result does not require unitarity or the existence of an invariant bilinear form (invariant in the sense of vertex algebras, see Section \ref{sec: invariant forms} for further discussion), and we do not require that the homogeneous subspaces for the grading by conformal dimensions be finite-dimensional.
There are many important examples of CFTs arising in mathematical and theoretical physics which require this level of generality, and in particular non-unitarity arises from the CFT-approach to classical critical phenomena, and from string theory.
Specific examples include the non-unitary Virasoro minimal models, affine vertex algebras at non-critical level (both universal and simple quotient), bosonic $\bbN$-graded affine W-algebras (again universal and simple quotient), and the $\beta \gamma$-ghost vertex algebra with central charge $c=2$ (along with other ``$A$-graded'' vertex algebras which arise in logarithmic conformal field theory).
As a result of our theorem there are canonical Wightman CFTs associated to these models, which demonstrates significant functional analytic regularity that is not otherwise apparent.

The constructions going from suitably defined Wightman CFTs to vertex algebras and back are given in Sections~\ref{sec: VOA to Wightman} and \ref{sec: Wightman to VOA}, respectively.
These are shown to give an equivalence of categories in Section~\ref{sec: equivalence of categories}.
The vertex algebra $\cV$ associated to a Wightman CFT with domain $\cD$ is constructed as a certain subspace $\cV \subset \cD$.
Conversely, the Wightman CFT associated to $\cV$ is constructed as an extension $\cV \subset \cD \subset \prod \cV(n)$, where $\prod \cV(n)$ is the algebraic completion (see Section \ref{sec: VOA to Wightman}).
We note that at this level of generality (allowing each weight space $\cV(n)$ to be infinite-dimensional) it is not a priori clear that there is a single Wightman CFT for each vertex algebra, and it seems plausible that there could be families of `Wightman completions' of a single vertex algebra.
However, as a consequence of our result, there is indeed a unique Wightman CFT for each M\"obius vertex algebra generated by quasiprimary fields.

A very useful and inspiring heuristic discussion on the connection between Wightman CFTs and vertex algebras can be found in \cite[\S1.2]{Kac98}. However, the arguments given there do not appear to be aimed to give precise mathematical details on this connection. 
More recently, three of the present authors gave a rigorous proof that \emph{unitary} M\"obius vertex algebras were equivalent to \emph{unitary} Wightman CFTs possessing an additional analytic property called \emph{uniformly bounded order}, provided that the homogeneous subspaces for the grading by conformal dimensions were finite-dimensional \cite{RaymondTanimotoTener22}.
The present article generalizes the previous result to possibly non-unitary theories, also dropping the requirements of uniformly bounded order and finite-dimensional weight spaces.
As the techniques historically used to study Wightman theories involve careful analysis of the norm topology on the space of states, there is significant new work required to generalize our previous results to the non-unitary setting.

We also demonstrate in Section~\ref{sec: invariant forms} that the correspondence constructed in this article is compatible with invariant bilinear forms, invariant sesquilinear forms, and invariant inner products.

\begin{thm*}
Let $\cD$ be a M\"obius-covariant Wightman CFT and let $\cV \subset \cD$ be the associated M\"obius vertex algebra.
Then every invariant inner product (unitary structure) on $\cD$ restricts to an invariant inner product on $\cV$, and conversely every unitary structure on $\cV$ uniquely extends to one on $\cD$.
The same holds for invariant sesquilinear forms (involutive structures) and invariant bilinear forms.
\end{thm*}

We are left with a striking and clear correspondence between two well-known axiomatizations of two-dimensional chiral conformal field theory, without any reliance on additional technical hypotheses.
We are motivated in part by the possibility to provide such an equivalence, which is not often possible in the wild landscape of axiomatic quantum field theory.
We are also motivated by intriguing links between non-unitary conformal field theories and the unitary world of algebraic conformal field theory.
Given a Wightman CFT on $S^1$ and an interval $I \subset S^1$, consider the algebra $\cP(I)$ generated by Wightman fields $\varphi(f)$ smeared by test functions $f$ supported in the interval $I$.
Such Wightman nets of algebras have been studied in the context of unitary quantum field theories \cite{StreaterWightman64}, and there is a substantial effort underway to understand the relationship between unitary vertex algebras, unitary Wightman nets, and the usual nets of algebras of \emph{bounded} observables (i.e.\! conformal nets) studied in algebraic conformal field theory \cite{CKLW18,GRACFT1}.
On the other hand, as a result of our present work, there exists a Wightman net for every M\"obius vertex algebra, including non-unitary ones.
Such nets could give an avenue to apply methods generally used in the unitary framework of algebraic quantum field theory in the more general setting of non-unitary models.
Previously such links have been probed only at the level of categories of representations \cite{EvansGannonNonunitary17}.
As a first demonstration of the potential of this approach we prove a version of the Reeh-Schlieder Theorem (regarding the cyclic and separating property of the vacuum vector) for non-unitary theories in Appendix~\ref{sec: reeh schlieder}.

Finally, there is strong motivation to understand functional analytic aspects of non-unitary vertex algebras as a part of studying links between algebraic and geometric aspects of the theory, as in \cite{HuangFunctionalI,HuangFunctionalII}.
More recently, analytic considerations of non-unitary vertex algebras have played a key role in the study of conformal blocks \cite{GuiConvergence24,GuiPropagation24,GuiAnalyticI23ax}, and such considerations also feature in the construction of functorial CFTs in the sense of Segal \cite{SegalDef}.

In future work, it would be interesting to relate modules for vertex algebras to representations of the corresponding Wightman nets, which would fit into the broad program underway in the unitary setting to relate vertex algebra modules to representations in algebraic conformal field theory \cite{TenerGRACFT2,TenerGRACFT3, gui21categoricalextension, gui20unbddax, CarpiWeinerXu}. Such relations should enable further correspondences between full two-dimensional conformal field theories in various approaches, cf.\! \cite{Moriwaki21, AGT23Pointed, AMT24OS}.

\subsection*{Acknowledgements}

S.C. and Y.T acknowledge support from the GNAMPA-INDAM project {\it Operator algebras and infinite quantum systems}, 
CUP E53C23001670001 and from the MUR Excellence Department Project {\it MatMod@TOV} awarded to the Department of Mathematics, University of Rome ``Tor Vergata'', CUP E83C23000330006. 
C.R. and J.T. were supported by ARC Discovery Project DP200100067, ``Physical realisation of enriched quantum symmetries''.
C.R. and J.T. would like to thank David Ridout for suggesting and encouraging our study of Wightman CFTs in the context of non-unitary models.

\section{Preliminaries on Wightman CFTs and M\"obius vertex algebras}

An \textbf{operator-valued distribution} on the unit circle $S^1$ with domain a vector space $\mathcal{D}$ is a linear map
\[
\varphi:C^\infty(S^1) \to \cL(\cD),
\]
where $\cL(\cD)$ is the space of linear operators on $\cD$.
In this article, we will typically study operator-valued distributions whose domain $\cD$ is infinite-dimensional, and we will require some topological considerations with respect to the action of sets of such distributions on $\cD$.

If $\cF$ is a set of operator-valued distributions on $S^1$ with a common domain $\cD$, then a linear functional $\lambda:\cD \to \bbC$ is called \textbf{compatible with $\cF$} if the multilinear maps $C^\infty(S^1)^k \to \bbC$ given by
\begin{equation}\label{eqn: preliminary regular action}
(f_1, \ldots, f_k) \mapsto \lambda\big(\varphi_1(f_1) \cdots \varphi_k(f_k)\Phi\big)
\end{equation}
are continuous in the $f_j$ for all $\varphi_1, \ldots, \varphi_k \in \cF$ and $\Phi \in \cD$.
Note that multilinear forms $C^\infty(S^1)^k \to \bbC$ are separately continuous if and only if they are jointly continuous since $C^\infty(S^1)$ is a Fr\'echet space \cite[Cor. \S 34.2]{Treves67}.
We write $\cD_\cF^*$ for the space of all linear functionals compatible with $\cF$.
Recall that a set $\cX$ of linear functionals on $\cD$ is said to separate points if for every non-zero $\Phi \in \cD$ there is a $\lambda \in \cX$ such that $\lambda(\Phi) \ne 0$.

\begin{defn}
A set $\cF$ of operator-valued distributions with domain $\cD$ acts \textbf{regularly} if $\cD_\cF^*$ separates points.
\end{defn}

If we imagine that $\cF$ consists of a family of Wightman fields (i.e.\! the operators $\varphi(f)$ are smeared quantum fields), then it is natural that functionals $\lambda:\cD \to \bbC$ should have the property that expressions \eqref{eqn: preliminary regular action} are continuous in the smearing functions $f_j$.
Thus, the condition of regularity serves to exclude certain nonphysical actions that have the property that functionals cannot distinguish states.
The following example illustrates the pathological behavior of nonregular actions.

\begin{ex}\label{ex: non regular action}
Let $\cD = T(C^\infty(S^1)) = \bigoplus_{k=0}^\infty C^\infty(S^1)^{\otimes k}$ be the tensor algebra, and for $f \in C^\infty(S^1)$ let $\varphi(f) \in \cL(\cD)$ be the operation of left-multiplication by $f$ in $\cD$.
The space $\cD$ carries a regular action of $\cF=\{\varphi\}$.
Let $\cI \subsetneq \cD$ be the left ideal generated by trigonometric polynomials $\bbC[z^{\pm 1}] \subset C^\infty(S^1)$.
Let $\tilde \cD = \cD/\cI$, and observe that for each $f \in C^\infty(S^1)$ the action of $\varphi(f)$ descends to an operator $\tilde \varphi(f) \in \cL(\tilde \cD)$.
The action of $\tilde \cF=\{\tilde \varphi\}$ on $\tilde \cD$ is not regular.
Let $\Omega \in \cD$ be the unit of the tensor algebra, and let $\tilde \Omega \in \tilde \cD$ be its image under the canonical projection.
For any $f \in \bbC[z^{\pm 1}]$ we have $\tilde \varphi(f)\tilde \Omega = 0$, and thus any $\lambda \in \tilde \cD_{\tilde \cF}^*$ vanishes on $\tilde \varphi(f)\tilde \Omega$ for any $f \in C^\infty(S^1)$.
In particular, if $f \in C^\infty(S^1) \setminus \bbC[z^{\pm 1}]$, then $\tilde \varphi(f)\tilde \Omega$ is non-zero but lies in the kernel of all $\lambda \in \tilde \cD_{\tilde \cF}^*$.
\end{ex}

\begin{rem}
A non-regular action of $\cF$ on $\cD$ descends to a regular action on the quotient $\cD/\cD_0$, where $\cD_0 = \bigcap_{\lambda \in \cD_\cF^*} \ker \lambda$.
\end{rem}

Let $e_n \in C^\infty(S^1)$ be the function $e_n(z)=z^n$.
The condition that $\cF$ acts regularly on $\cD$ ensures that the operators $\varphi(f)$ are determined by the modes $\varphi(e_n)$ in a certain sense that we will make precise below.
This is in contrast with Example~\ref{ex: non regular action}, in which $\tilde \varphi(e_n)\tilde \Omega = 0$ for all $n$ but $\tilde \varphi(f)\tilde \Omega \ne 0$ for some $f \in C^\infty(S^1)$.

We now introduce certain topologies on $\cD$ associated with the action of $\cF$.
We assume here that the reader is familiar with (or indifferent to) the fundamentals of topological and locally convex vector spaces, and defer the relevant background and additional details to Appendix~\ref{app: TVS and LCS}.

\begin{defn}
Given a family $\cF$ of operator-valued distributions on $S^1$ with domain $\cD$, the \textbf{$\cF$-weak} topology on $\cD$ is the weak topology induced by the linear functionals $\cD_\cF^*$.
That is, the $\cF$-weak topology is the coarsest topology such that every $\lambda \in \cD_\cF^*$ is continuous.
\end{defn}

For a topological vector space $X$, a map $T:X \to \cD$ is $\cF$-weakly continuous precisely when $\lambda \circ T$ is continuous for all $\lambda \in \cD_\cF^*$.
The family $\cF$ acts regularly precisely when the $\cF$-weak topology is Hausdorff.
We will see in Lemma~\ref{lem: continuity of field operators} below that for $\varphi \in \cF$ the expressions $\varphi(f)\Phi$ are continuous in $f \in C^\infty(S^1)$ when $\cD$ is given the $\cF$-weak topology, so indeed $\varphi(f)$ is determined by the modes $\varphi(e_n)$ when $\cF$ acts regularly.

There is a second natural topology on $\cD$ associated with the action of $\cF$.
For $\varphi_1, \ldots, \varphi_k \in \cF$ and $\Phi \in \cD$, let
\[
S_{\varphi_1,\ldots,\varphi_k,\Phi}:C^\infty(S^1)^{\otimes k} \to \cD
\]
be the linear map
\[
S_{\varphi_1,\ldots,\varphi_k,\Phi}(f_1\otimes \cdots \otimes f_k) = \varphi_1(f_1) \cdots \varphi_k(f_k)\Phi.
\]
We equip the algebraic tensor product $C^\infty(S^1)^{\otimes k}$ with the projective topology, for which continuous linear maps $C^\infty(S^1)^{\otimes k} \to X$ correspond to continuous multilinear maps (see Appendix~\ref{app: TVS and LCS}).

\begin{defn}
Given a family of $\cF$ of operator-valued distributions on $S^1$ with domain $\cD$, the \textbf{$\cF$-strong} topology on $\cD$ is the colimit (or final) locally convex topology induced by the maps $S_{\varphi_1,\ldots,\varphi_k,\Phi}$ for $\varphi_1, \ldots, \varphi_k \in \cF$ and $\Phi \in \cD$.
That is, the $\cF$-strong topology is the finest locally convex topology such that the maps $S_{\varphi_1,\ldots,\varphi_k,\Phi}$ are continuous.
\end{defn}

Equivalently, the $\cF$-strong topology is the finest locally convex topology on $\cD$ such that expressions
$
\varphi_1(f_1) \cdots \varphi_k(f_k)\Phi
$
are continuous in the functions $f_j$ (jointly, or equivalently separately by \cite[Cor. \S34.2]{Treves67}).

\begin{rem}\label{rem: weak of F strong is F weak}
For a locally convex space $X$, a linear map $T:\cD \to X$ is continuous precisely when $T \circ S_{\varphi_1, \ldots, \varphi_k, \Phi}$ is continuous for all $\varphi_1, \ldots, \varphi_k \in \cF$ and $\Phi \in \cD$ \cite[Thm. 12.2.2]{NariciBeckenstein11}.
In particular, a linear functional $\lambda:\cD \to \bbC$ is $\cF$-strongly continuous if and only if $\lambda \in \cD_\cF^*$, and so the weak topology on $\cD$ induced by the space of $\cF$-strong continuous linear functionals is precisely the $\cF$-weak topology.
\end{rem}

We now have the following alternate characterizations of the regularity of an action of $\cF$ on $\cD$.

\begin{lem}\label{lem: characterization of regular action}
Let $\cF$ be a set of operator-valued distributions on $S^1$ with domain a vector space $\cD$.
Then the following are equivalent.
\begin{enumerate}[i)]
\item $\cF$ acts regularly on $\cD$, i.e.\! $\cD_\cF^*$ separates points.
\item The $\cF$-weak topology on $\cD$ is Hausdorff.
\item The $\cF$-strong topology on $\cD$ is Hausdorff.
\item There exists a locally convex Hausdorff topology on $\cD$ such that the maps
\[
(f_1, \ldots, f_k) \mapsto \varphi_1(f_1) \cdots \varphi_k(f_k)\Phi
\]
are continuous $C^\infty(S^1)^k \to \cD$ for all $\varphi_1, \ldots, \varphi_k \in \cF$ and $\Phi \in \cD$.
\end{enumerate}
\end{lem}
\begin{proof}
As noted above, the implication (i) $\implies$ (ii) follows immediately from the definitions of regularity and the $\cF$-weak topology.
The identity map $\cD \to \cD$ is continuous from the $\cF$-weak topology to the $\cF$-strong topology, and thus (ii) $\implies$ (iii).
The $\cF$-strong topology is locally convex by definition, and thus (iii) $\implies$ (iv) is tautological.
Finally, if $\tau$ is a locally convex Hausdorff topology on $\cD$ as in (iv), then we have an inclusion of continuous duals $(\cD,\tau)^* \subset \cD_\cF^*$.
By the Hahn-Banach theorem $(\cD,\tau)^*$ separates points \cite[Thm. 7.7.7]{NariciBeckenstein11}. Hence so does $\cD_\cF^*$, and the action of $\cF$ is regular.
\end{proof}

Both the $\cF$-strong and $\cF$-weak topologies are quite natural, and so it is not surprising that the fields $\cF$ act continuously when $\cD$ is given one of these topologies.

\begin{lem}\label{lem: continuity of field operators}
Let $\cF$ be a set of operator-valued distributions on $S^1$ acting regularly with domain $\cD$  equipped with the $\cF$-strong topology.
Then for $\varphi \in \cF$ the natural map $\varphi:C^\infty(S^1) \times \cD \to \cD$ is  separately continuous.
The same holds if $\cD$ is equipped with the $\cF$-weak topology.
\end{lem}
\begin{proof}
First, fix $f \in C^\infty(S^1)$.
For any $\varphi_1, \ldots, \varphi_k \in \cF$ and $\Phi \in \cD$, the expression 
\[
\varphi(f)\varphi_1(f_1) \cdots \varphi_k(f_k)\Phi
\] 
is continuous in the $f_k$ by the definition of the $\cF$-strong topology, and thus $\varphi(f):\cD \to \cD$ is continuous by Remark~\ref{rem: weak of F strong is F weak}.
Similarly, for fixed $\Phi \in \cD$ expressions $\varphi(f)\Phi$ are continuous in $f$, and we conclude that $\varphi$ is separately continuous.

Now consider if $\cD$ is given the $\cF$-weak topology.
For $\lambda \in \cD_\cF^*$ the expression $\lambda(\varphi(f)\Phi)$ is evidently continuous in $f$, and it just remains to show that $\varphi(f)$ acts continuously on $(\cD,\cF\mbox{-weak})$.
Let $\cD^\sharp$ be the algebraic dual of $\cD$, and let $\varphi(f)^*:\cD^\sharp \to \cD^\sharp$ be the adjoint action.
If $\lambda \in \cD_\cF^*$ then 
\begin{equation}\label{eqn: adjoint action of varphif}
(\varphi(f)^*\lambda)(\varphi_1(f_1) \cdots \varphi_k(f_k)\Phi) = \lambda(\varphi(f)\varphi_1(f_1) \cdots \varphi_k(f_k)\Phi)
\end{equation}
is continuous in the functions $f_j$, so $\varphi(f)^*$ leaves $\cD_\cF^*$ invariant.
Thus if $\Phi_n$ is a net in $\cD$ converging $\cF$-weakly to $\Phi$ then 
\[
\lambda(\varphi(f)\Phi_n) = (\varphi(f)^*\lambda)\Phi_n \to (\varphi(f)^*\lambda)\Phi = \lambda(\varphi(f)\Phi).
\]
Hence $\varphi(f)\Phi_n$ converges $\cF$-weakly to $\varphi(f)\Phi$ and $\varphi(f)$ acts continuously on $\cD$.
\end{proof} 

We now assume our vector space $\cD$ is equipped with a family $\cF$ of operator-valued distributions on $S^1$ as well as a representation $U:\Mob \to \End(\cD)$, where $\Mob=\mathrm{PSU}(1,1)$ is the group of holomorphic automorphisms of the closed unit disk. We mention here for clarity that we will use the notation $U$ for the representation throughout the paper, in particular, for when the theories are not unitary. As this notation is commonly used for representations of the global symmetry group we continue to use it as the setting of our paper is clear. 
As in \cite[\S 6]{CKLW18}, for $\gamma \in \Mob$ we denote by $X_\gamma \in C^\infty(S^1)$ the function
\[
X_\gamma(\rme^{\rmi\vartheta}) = -\rmi \frac{d}{d \vartheta} \log(\gamma(\rme^{\rmi \vartheta})),
\]
which takes positive real values since $\gamma$ is an orientation-preserving diffeomorphism of $S^1$.
For $f \in C^\infty(S^1)$ and $d \in \bbZ_{\ge 0}$ we denote by $\beta_d(\gamma)f \in C^\infty(S^1)$ the function
\begin{equation}\label{eq:beta}
(\beta_d(\gamma)f)(z) = (X_\gamma(\gamma^{-1}(z)))^{d-1} f(\gamma^{-1}(z)).
\end{equation}
An operator-valued distribution with domain $\cD$ is called \textbf{M\"obius-covariant with conformal dimension $d$} under the representation $U$
if for every $\gamma \in \Mob$ and every $f \in C^\infty(S^1)$ we have 
\[
U(\gamma)\varphi(f)U(\gamma)^{-1} = \varphi(\beta_d(\gamma)f)
\]
as endomorphisms of $\cD$.
We say that a vector $\Phi \in \cD$ has \textbf{conformal dimension $d \in \bbZ$} if $U(R_\vartheta)\Phi = \rme^{\rmi d\vartheta}\Phi$ for all rotations $R_\vartheta \in \Mob$.

We now present the not-necessarily-unitary version of the Wightman axioms for two-dimensional chiral conformal field theories on the circle $S^1$.
Historically, the Wightman axioms have been closely entwined with \emph{unitary} theories, where the space of states $\cD$ possess an appropriate inner product.
Non-unitary versions of the Wightman axioms have also appeared in various contexts such as the mathematical description of gauge fields, see e.g.\! \cite[\S 6.4]{Strocchi93}. 
In this article we will generally refer to the non-unitary theories in question simply as Wightman conformal field theories for the sake of brevity.

\begin{defn}\label{def: Wightman}
Let $\cD$ be a vector space equipped with a representation $U$ of $\Mob$ and a choice of non-zero vector $\Omega \in \cD$.
Let $\mathcal{F}$ be a set of operator-valued distributions on $S^1$ acting regularly on their common domain $\mathcal{D}$.
This data forms a (not-necessarily-unitary) \textbf{M\"obius-covariant Wightman CFT} on $S^1$ if they satisfy the following axioms:
\begin{enumerate}[{(W}1{)}]
\item \textbf{M\"obius covariance}:  For each $\varphi \in \mathcal{F}$ there is $d \in \mathbb{Z}_{\ge 0}$ such that $\varphi$ is M\"obius-covariant with conformal dimension $d$ under the representation $U$.\label{itm: W Mob covariance}
\item \textbf{Locality}: If $f$ and $g$ have disjoint supports, then $\varphi_1(f)$ and $\varphi_2(g)$ commute
for any pair $\varphi_1, \varphi_2 \in \mathcal{F}$.
\label{itm: W locality}
\item \textbf{Spectrum condition}: If $\Phi \in \cD$ has conformal dimension $d < 0$ then $\Phi = 0$.
\label{itm: W spec}
\item \textbf{Vacuum}: The vector $\Omega$ is invariant under $U$, and $\cD$ is spanned by vectors of the form $\varphi_1(f_1) \cdots \varphi_k(f_k)\Omega$.
\label{itm: W Vacuum}
\setcounter{wightman}{\value{enumi}}
\end{enumerate}
\end{defn}

A M\"obius-covariant Wightman CFT is a quadruple $(\cF,\cD,U,\Omega)$, but we will frequently refer to the family $\cF$ of fields or the domain $\cD$ as a (M\"obius-covariant) Wightman CFT when the remaining data is clear from context.

Let $e_j(z) = z^j \in C^\infty(S^1)$, and let 
\[
\cV(n) = \spann \{\varphi_1(e_{j_1}) \cdots \varphi_k(e_{j_k})\Omega \, | \, j_1 + \cdots + j_k = -n\}.
\]

By M\"obius covariance (W\ref{itm: W Mob covariance}) the vectors in $\cV$ have conformal dimension $n$, or in other words $U(R_{\vartheta})$ acts on $\cV(n)$ as multiplication by $\rme^{\rmi n \vartheta}$.

\begin{lem}\label{lem:V(n) dense}
Let ($\cF$,$\cD$,$U$,$\Omega$) be a M\"obius-covariant Wightman CFT, and suppose $\cD$ is equipped with the $\cF$-strong topology.
\begin{enumerate}[i)]
\item The map $U:\Mob \times \cD \to \cD$ is separately continuous.
\item  $\bigoplus_{n \ge 0} \cV(n)$ is dense in $\cD$.
\end{enumerate}
The same holds for the $\cF$-weak topology.
\end{lem}
\begin{proof}
First consider the $\cF$-strong topology.
Fix $\Phi \in \cD$, and we will show that $U(\gamma)\Phi$ is continuous in $\gamma$.
By the vacuum axiom we may assume without loss of generality that $\Phi = \varphi_1(f_1) \cdots \varphi_k(f_k)\Omega$.
We have
\[
U(\gamma)\varphi_1(f_1) \cdots \varphi_k(f_k)\Omega = \varphi_1(\beta_{d_1}(\gamma)f_1) \cdots \varphi_k(\beta_{d_k}(\gamma)f_k)\Omega.
\]
The smooth function $\beta_d(\gamma)f$ depends continuously on $\gamma$, and thus  $U(\gamma)\varphi_1(f_1) \cdots \varphi_k(f_k)\Omega$ depends continuously on $\gamma$ as well.
Now consider a fixed $\gamma \in \Mob$, and by the same argument $U(\gamma)\varphi_1(f_1) \cdots \varphi_k(f_k)\Omega$ depends continuously on the $f_j$.
Hence by Remark~\ref{rem: weak of F strong is F weak} and the vacuum axiom $U(\gamma)\Phi$ depends continuously on $\Phi$, proving (i).
The argument is similar for the $\cF$-weak topology.

For (ii), note that $\bigoplus_{n \in \bbZ} \cV(n)$ is dense since Laurent polynomials are dense in $C^\infty(S^1)$, and $\cV(n) = 0$ for $n < 0$ by the spectrum condition.
\end{proof}

Let $\cD_\cF^* \cap \cV(n)^*$ denote the space of linear functionals $\lambda \in \cD_\cF^*$ such that $\lambda|_{\cV(m)} = 0$ when $m\ne n$.
The following technical observations will be essential in constructing vertex algebras from Wightman CFTs.

\begin{lem}\label{lem: pos eigenvalues dense implies spectrum condition}
Suppose that ($\cF$,$\cD$,$U$,$\Omega$) satisfies all of the axioms of a M\"obius-covariant Wightman CFT except perhaps for the spectrum condition.
\begin{enumerate}[i)]
\item \label{itm: Vnstar separates points} For $n \in \bbZ$, $\cD_\cF^* \cap \cV(n)^*$ separates points in $\cV(n)$.
\item \label{itm: dense implies positive energy} If $\bigoplus_{n \ge 0} \cV(n)$ is $\cF$-strongly (or $\cF$-weakly) dense then the spectrum condition holds.
\end{enumerate}
\end{lem}
\begin{proof}
First we prove (\ref{itm: Vnstar separates points}).
From the definition of a Wightman CFT, $\cD_\cF^*$ separates points so given $v \in \cV(n)$ we may choose $\lambda \in \cD_\cF^*$ such that $\lambda(v)\ne 0$.
For $z=\rme^{\rmi \vartheta}$, let $r_z = U(R_\vartheta) \in \Mob$ be rotation by $z \in S^1$.
If $r_z^*$ is the adjoint operator, then $r_z^*\lambda \in \cD_\cF^*$.
Let $\lambda_n: \cD \to \bbC$ be given by
\[
\lambda_n(\Phi) = \frac{1}{2\pi\rmi}\int_{S^1} z^{-n-1} \bil{\Phi, r_z^*\lambda}_{\cD,\cD_\cF^*} \, dz.
\]
To see that the integral exists, observe that
\[
\bil{\varphi_1(f_1) \cdots \varphi_k(f_k)\Omega, r_z^*\lambda}_{\cD,\cD_\cF^*}
=
\bil{\varphi_1(\beta_{d_1}(r_z)f_1) \cdots \varphi_k(\beta_{d_k}(r_z)f_k)\Omega,\lambda}_{\cD,\cD_\cF^*}.
\]
Since $(z,f) \mapsto \beta_d(r_z)f$ is jointly continuous $S^1 \times C^\infty(S^1) \to C^\infty(S^1)$, and $\lambda \in \cD_\cF^*$, the expression $\bil{\varphi_1(f_1) \cdots \varphi_k(f_k)\Omega, r_z^*\lambda}$ is jointly continuous $S^1 \times C^\infty(S^1)^k \to \bbC$, and thus the integral defining $\lambda_n$ exists.
Moreover, we see that $\lambda \in \cD_\cF^*$.
If $u \in \cV(m)$, then $\lambda_n(u)=\delta_{n,m} \lambda(u)$, so $\lambda_n \in \cD_\cF^* \cap \cV(n)^*$ and $\lambda_n(v)=\lambda(v) \ne 0$.

For part (ii), it suffices to consider the $\cF$-strong topology.
Let $\cW(n) \subset \cD$ be the subspace of vectors with conformal dimension $n$.
Note that 
\[
\cV(n) = \spann \{\varphi_1(e_{j_1}) \cdots \varphi_k(e_{j_k})\Omega \, \vert \, \sum j_i = -n\} \subset \cW(n)
\]
but equality is not immediate (it is e.g.\! a consequence of Proposition~\ref{prop: canonical embedding of D}).
Now fix $n<0$ and let $v \in \cW(n)$.
To verify the spectrum condition we must show that $v=0$.
Let $\lambda \in \cD_\cF^*$, and as above let $\lambda_n = \frac{1}{2\pi\rmi}\int_{S^1} z^{-n-1} r_z^* \lambda \, dz$.
Then $\lambda_n(v)=\lambda(v)$, but $\lambda_n$ vanishes on $\bigoplus_{m \ge 0} \cV(m)$.
Since the latter space is assumed to be dense, and $\lambda_n \in \cD_\cF^*$ is continuous, we have $\lambda_n \equiv 0$.
Hence $\lambda(v) = 0$, and since $\cD_\cF^*$ separates points we have $v=0$, as desired.
\end{proof}

We will later see that for every $\lambda \in \cV(n)^*$ there is a unique extension to a linear functional in $\cD_\cF^*$ that vanishes on $\cV(m)$ (when $m \ne n$) -- see Proposition~\ref{prop: canonical embedding of D}.

The notion of Wightman CFT presented in Definition \ref{def: Wightman} generalizes the unitary notion of Wightman CFT considered in \cite{RaymondTanimotoTener22} in several ways.
Most notably, the domain $\cD$ does not have an inner product.
We also do not assume that $\Omega$ is the unique $\Mob$-invariant vector up to scale, and we do not require that the eigenspaces for the generator of rotation are finite-dimensional.

The requirement that $\cF$ act regularly on $\cD$ is necessary, as there exist pathological examples of quadruples $(\cF,\cD,U,\Omega)$ that satisfy all of the requirements to be a Wightman CFT except for the regularity of the action of $\cF$.
Indeed, we can refine Example~\ref{ex: non regular action} to produce a quadruple with non-regular action for which the only finite-energy vectors are scalar multiples of the vacuum, despite $\cD$ being infinite-dimensional.

\begin{ex}\label{ex non regular Wightman CFT}
Let $\cX_{< 0} \subset C^\infty(S^1)$ be the closed span of $z^{-1}, z^{-2}, \ldots$.
These are the functions $f$ in $C^\infty(S^1)$ that extend to holomorphic functions outside the unit disk which vanish at infinity.
We similarly define $\cX_{\ge 0}$.
Let $p:C^\infty(S^1) \to \cX_{<0}$ be the projection with kernel $\cX_{\ge 0}$.
Let $\cD = S(\cX_{<0}) = \bigoplus_{k=0}^\infty S^k(\cX_{<0})$ be the symmetric algebra, and let $\Omega \in S^0(\cX_{<0})$ be the unit.
Let $\varphi$ be the operator-value distribution with domain $\cD$ where $\varphi(f)$ acts by multiplication by $pf$ in $S(\cX_{<0})$.
This action is evidently regular.
The family $\cF=\{\varphi\}$ acts locally on $\cD$ since the symmetric algebra is abelian, and the vacuum vector is cyclic for $\cF$.

Define a representation of $\Mob$ on $\cX_{<0}$ by $U(\gamma)f = p(f \circ \gamma^{-1})$, and extend this representation to $S(\cX_{<0})$ (this representation is better understood as the quotient of the natural representation of $\Mob$ on $\cX_{\le 0}$ by constant functions).
The representation has positive energy, and $\Omega$ is $\Mob$-invariant.
Moreover we have for $f \in C^\infty(S^1)$
\[
U(\gamma)p f = p ((pf) \circ \gamma^{-1}) = p (f \circ \gamma^{-1}) \quad\text{(as the constant component is annihilated by } p\text{)}
\]
and it follows that $U(\gamma) \varphi(f) = \varphi(f\circ \gamma^{-1})U(\gamma) = \varphi(\beta_1(\gamma)f)U(\gamma)$.
Hence $\varphi$ is M\"obius covariant with conformal dimension 1, and we have shown that $(\cF,\cD,U,\Omega)$ is a Wightman CFT.

Now let $\cX_\omega \subset C^\infty(S^1)$ be the dense subspace of functions which extend holomorphically to a neighborhood of $S^1$.
Then $\cX_\omega$ is invariant under $U(\gamma)$.
Moreover, a function $f$ lies in $\cX_\omega$ precisely when its Fourier coefficients decay sufficiently rapidly, and thus $\cX_\omega$ is invariant under $p$ as well.
Let $\cI \subsetneq \cD$ be the left ideal generated by $p\cX_\omega  \subset \cX_{< 0}$, and observe that $\cI$ is invariant under $\varphi(f)$ and $U(\gamma)$ for all $f \in C^\infty(S^1)$ and $\gamma \in \Mob$.
Let $\tilde \cD = \cD/\cI$, let $\tilde \Omega \in \tilde \cD$ be the image of $\Omega$ under the canonical projection, and let $\tilde \varphi(f)$ and $\tilde U(\gamma)$ be the operators on $\tilde \cD$ induced by $\varphi(f)$ and $U(\gamma)$, respectively.
The quadruple $(\tilde \cF, \tilde \cD, \tilde U, \tilde \Omega)$ satisfy all of the requirements of a Wightman CFT except for regularity of the action of $\cF$.
We have 
\[
\cV(n) = \spann\{\tilde \varphi_1(e_{j_1}) \cdots \tilde \varphi_k(e_{j_k})\tilde \Omega \mid j_1 + \cdots + j_k = -n\} = \{0\}
\]
for all non-zero $n \in \bbZ$ since $e_j \in \cX_\omega$.
\end{ex}

\begin{rem}
If $(\cF,\cD,U,\Omega)$ satisfy all of the requirements of a Wightman CFT except that the fields $\cF$ do not act regularly, then we can obtain a Wightman CFT on a quotient of $\cD$.
Let $\cD_0=\bigcap_{\lambda \in \cD_\cF^*} \operatorname{ker} \lambda$.
Since $\cD_\cF^*$ is invariant under the adjoint actions $\varphi(f)^*$ and $U(\gamma)^*$, it follows that $\cD_0$ is invariant under $\varphi(f)$ and $U(\gamma)$, and so we have actions of smeared fields and $\Mob$ on the quotient $\tilde \cD = \cD/\cD_0$.
So long as $\cD_0 \ne \cD$, these actions give a Wightman CFT on $\tilde \cD$.

If one applies this procedure to the example $(\tilde \cF, \tilde \cD, \tilde U, \tilde \Omega)$ constructed in Example \ref{ex non regular Wightman CFT}, then one obtains the trivial Wightman CFT.
Indeed, we can check that $\tilde\cD = \bbC \tilde \Omega \oplus \tilde\cD_0$ as follows.
Returning to the notation of Example~\ref{ex non regular Wightman CFT}, we have $\cI = \bigoplus_{k=1}^\infty \cI\cap S^k(\cX_{<0})$, and thus 
\[
\tilde \cD = \bbC \tilde \Omega \oplus \bigoplus_{k=1}^\infty S^k(\cX_{<0})/(\cI \cap S^k(\cX_{<0})).
\]
Hence it suffices to check that an arbitrary $\lambda \in \tilde \cD_{\tilde \cF}^*$ vanishes on each of the spaces $S^k(\cX_{<0})/(\cI \cap S^k(\cX_{<0}))$ when $k \ge 1$.
Note that this space is spanned by vectors of the form $\tilde \varphi(f_1)\cdots\tilde \varphi(f_k)\Omega$.
When one of the functions $f_j$ lies in $\cX_\omega$ we have $\varphi(f_1)\cdots \varphi(f_k)\Omega  \in \cI$ and thus $\tilde \varphi(f_1)\cdots\tilde \varphi(f_k)\Omega = 0$.
Hence for any $\lambda \in \tilde \cD_{\tilde \cF}^*$ we have
\begin{equation}\label{eqn: lambda tilde phi}
\lambda(\tilde \varphi(f_1)\cdots\tilde \varphi(f_k)\Omega) = 0
\end{equation}
when some $f_j$ lies in $\cX_\omega$.
Since $\cX_\omega$ is dense in $C^\infty(S^1)$ and $\lambda \in \tilde \cD_{\tilde \cF}^*$, it follows that \eqref{eqn: lambda tilde phi} holds for arbitrary functions $f_j \in C^\infty(S^1)$, and we conclude that $\lambda$ vanishes on $S^k(\cX_{<0})/(\cI \cap S^k(\cX_{<0}))$ when $k \ge 1$.
Hence $\tilde\cD = \bbC \tilde \Omega \oplus \tilde\cD_0$ as claimed, and $\tilde{\tilde \cD}$ is spanned by the vacuum $\tilde{\tilde \Omega}$.
\end{rem}

With these distinctions in mind, we give the corresponding notion of vertex algebra after some brief preliminaries.
Let $\Lie(\Mob)$ be the three-dimensional real Lie algebra of $\Mob=\mathrm{PSU}(1,1)$.
If $\Mob$ is regarded as a subgroup of the group $\mathrm{Diff}(S^1)$ of orientation-preserving diffeomorphisms of the unit circle $S^1$, then $\Lie(\Mob)$ is identified with a three-dimensional subspace of the space of smooth vector fields $\mathrm{Vect}(S^1)$ on $S^1$.
Each vector field is identified with a differential operator $f(\rme^{\rmi\vartheta})\frac{d}{d\vartheta}$ for some smooth function $f(\rme^{\rmi\vartheta})$,
and the Lie bracket is given by $[f\frac{d}{d\vartheta},g\frac{d}{d\vartheta}] = (f'g - fg')\frac{d}{d\vartheta}$, where $f'$ denotes $\tfrac{df}{d\vartheta}$.
Note that this bracket is the opposite of the bracket of vector fields, which is the natural choice when identifying $\mathrm{Vect}(S^1)$ with the Lie algebra of $\mathrm{Diff}(S^1)$.
The complexification $\Lie(\Mob)_\bbC \cong \mathfrak{sl}(2,\bbC)$ of $\Lie(\Mob)$ is spanned by the elements $\{L_{-1},L_0,L_1\}$,
where $L_m$ is the complexified vector field $-\rmi \rme^{\rmi m\vartheta}\frac{d}{d\vartheta}$.
The vector fields $L_m$ satisfy the commutation relations
\begin{align*}
 [L_m, L_n] = (m-n)L_{m+n}, \qquad m,n=-1,0,1.
\end{align*}
In a representation of $\Lie(\Mob)_\bbC$, we will frequently abuse notation and write $L_k$ for the operator corresponding to the vector field indicated above.

If $\cV$ is a vector space we write $\End(\cV)[[z^{\pm 1}]]$ for the vector space of formal power series in $z^{\pm 1}$ with coefficients in $\End(\cV)$.
Given $v \in \cV$ and $A(z) =\sum_{n \in \bbZ} A_n z^n \in \End(\cV)[[z^{\pm 1}]]$, we have a formal series $A(z)v = \sum_n A_n v z^n$ with coefficients in $\cV$. For any $B \in \End(\cV)$ we have
$[A(z), B] = \sum_n [A_n, B]z^n \in \End(\cV)[[z^{\pm 1}]]$. If $B(w)$ is another formal series in a second formal variable $w$,
then the expression $[A(z), B(w)]$ makes sense as a formal series in $z^{\pm 1}$ and $w^{\pm 1}$.

We can now precisely specify the flavor of vertex algebras that we will consider.%
\footnote{The term M\"obius vertex algebra has been used in the literature to describe various slightly different notions (see e.g.\! \cite{BuhlKaraali08, HLZI,HuangGenerators20,Kac98}).
In some cases, authors include the possibility of fermionic fields, with the corresponding super version of the locality axiom; the term M\"obius vertex superalgebra is also used in this case.
Additionally, some authors replace our $\bbN$-grading with a more general grading. 
We do not foresee any significant obstacles to generalizing our results to M\"obius vertex superalgebras graded by a lower-bounded subset of $\tfrac12\bbZ$.}

\begin{defn}
An ($\bbN$-graded) \textbf{M\"obius vertex algebra} consists of a vector space $\cV$ equipped with a representation $\{L_{-1},L_0,L_1\}$ of $\Lie(\Mob)_\bbC$, a state-field correspondence $Y:\cV \to \End(\cV)[[z^{\pm 1}]]$,
and a choice of non-zero vector $\Omega \in \cV$ such that the following hold:
\begin{enumerate}[{(VA}1{)}]
\item $\cV = \bigoplus_{n=0}^\infty \cV(n)$, where $\cV(n) =  \ker(L_0 - n)$.
\item $Y(\Omega,z) = \operatorname{Id}_\cV$ and $\left.Y(v,z)\Omega\right|_{z=0} = v$, i.e. $Y(v,z)\Omega$ has only non-negative powers of $z$ for all $v \in \cV$.
\item $\Omega$ is $\Lie(\Mob)$-invariant, i.e.\! $L_m \Omega = 0$ for $m=-1,0,1$.
\item $[L_m, Y(v,z)] = \sum_{j=0}^{m+1} \binom{m+1}{j}z^{m+1-j} Y(L_{j-1}v,z)$ and $Y(L_{-1}v,z) = \frac{d}{dz} Y(v,z)$ for all $v \in \cV$ and $m = -1,0,1$.
\item For each $u,v \in \cV$, there exists $N$ sufficiently large such that \newline $(z-w)^N[Y(v,z), Y(u,w)] = 0$.
\end{enumerate}
\end{defn}

An immediate consequence of (VA3) is that $\Omega \in \cV(0)$. For $v \in \cV$, we write $Y(v,z) = \sum_{m \in \bbZ} v_{(m)}z^{-m-1}$, where $v_{(m)} \in \End(\cV)$ are called the modes of $v$.
A vector $v \in \cV$ is called \textbf{homogeneous (with conformal dimension $d$)} if it lies in $\cV(d)$. As a consequence of the $L_0$-commutation relation, when $v$ is homogeneous with conformal dimension $d$ we have $[L_0, v_{(m)}] = (d-m-1)v_{(m)}$. Hence $v_{(m)}$ maps $\cV(n)$ into $\cV(n+d-m-1)$.
A vector $v \in \cV$ is called \textbf{quasiprimary} if it is homogeneous and $L_{1}v=0$. A field $Y(v,z)$ is called quasiprimary if its corresponding vector $v$ is a quasiprimary vector.

We say that a M{\"o}bius vertex algebra $\cV$ is generated by a (possibly infinite) set of vectors $S$ if $\cV$ is spanned by monomials in the modes of only those vectors acting on $\Omega$. Note that we do note make assumptions on the modes $v_{(n)}$, $v \in S$, such as requiring that $n \leq -1$. As the axioms give a correspondence between states and fields, one can equivalently refer to a M{\"o}bius vertex algebra as being generated by a set of fields if the modes of a set of fields generates in the above sense.

We record two useful identities satisfied by the modes of a vertex operator (see \cite[\S4.8]{Kac98}), the \textbf{Borcherds product formula}:
\begin{equation}\label{eqn: BPF}
\big(u_{(n)}v\big)_{(k)}= \sum_{j=0}^\infty (-1)^j \binom{n}{j} \left(u_{(n-j)}v_{(k+j)} - (-1)^{n}v_{(n+k-j)}u_{(j)}\right),
\end{equation}
and the \textbf{Borcherds commutator formula}:
\begin{equation}\label{eqn: BCF}
[u_{(m)},v_{(k)}] = \sum_{j =0}^\infty \binom{m}{j} \big(u_{(j)}v\big)_{(m+k-j)}.
\end{equation}
Note that when the sums of operators on the right-hand sides of \eqref{eqn: BPF} and \eqref{eqn: BCF} are applied to a vector, all but finitely many terms vanish.

\section{Equivalence between M\"obius vertex algebras and Wightman CFTs}\label{sec: MVA and WCFT}

\subsection{From vertex algebras to Wightman CFTs}\label{sec: VOA to Wightman}

In this section we construct a Wightman CFT from a M\"obius vertex algebra $\cV=\bigoplus_{n=0}^\infty \cV(n)$.
The first step is to construct operator-valued distributions from the formal distributions $Y(v,z)$, as follows.
For $v \in \cV(d)$, the degree-shifted mode $v_n$ is defined by $v_n:=v_{(n+d-1)}$, which gives an alternative field expansion $Y(v,z)=\sum_{n=-\infty}^\infty v_{n}z^{-n-d}$, so that $v_n \cV(m) \subset \cV(m-n)$.
We extend the definition of $v_n$ to non-homogeneous vectors by linearity.
Let us write $\cV'$ for the restricted dual $\cV' = \bigoplus_{n=0}^\infty \cV(n)^*$, which is to say linear functionals on $\cV$ that are supported on finitely many $\cV(n)$. We note that this is strictly smaller than the algebraic dual $\cV^{*}$.
We denote by $\widehat \cV$ the algebraic completion, defined as a direct product
\[
\widehat \cV := \prod_{n=0}^\infty \cV(n),
\]
and we embed $\cV \subset \widehat \cV$ in the natural way.
We equip $\widehat \cV$ with the weak topology induced by the pairing with $\cV'$.

For $f \in C^\infty(S^1)$ we define 
\[
Y^0(v,f):\cV \to \widehat \cV
\]
by 
\[
Y^0(v,f)u = \sum_{n \in \bbZ} \hat f(n) v_n u,
\]
where $\hat f(n)$ is the $n$-th Fourier coefficient of $f$.

We now show that the maps $Y^0(v,f)$ may be extended to act on an invariant domain $\cD \subset \widehat \cV$.
The first step is the following lemma, which has an identical proof to \cite[Lem.\! 2.7]{RaymondTanimotoTener22}.

\begin{lem}
For all $v^1, \ldots, v^k,u \in \cV$ and $u' \in \cV'$, there exists a polynomial $p$ such that 
$$
\big|\bil{v^1_{m_1} v^2_{m_2} \cdots v^k_{m_k}u, u^\prime}_{\cV,\cV'}\big| \le \abs{p(m_1, \ldots, m_k)}
$$
for all $(m_1, \ldots, m_k) \in \bbZ^k$.
The polynomial depends on the vectors $v^j,u,$ and $u^\prime$, but the degree of $p$ is bounded from above by a number that is independent of $u$ and $u'$.
\end{lem}

We refer the reader to \cite{RaymondTanimotoTener22} for a proof.
A very similar argument, however, yields the following observation.

\begin{lem}\label{lem: va subspace finite dimensional}
Let $v^1, \ldots, v^k \in \cV$ and let $n \in \bbZ_{\ge 0}$.
Then
\[
\cV(n; v^1, \cdots, v^k) := \spann \{ v^1_{m_1} \cdots v_{m_k}^k \Omega \, | \, m_1 + \cdots + m_k = -n\}
\]
is finite-dimensional.
\end{lem}
\begin{proof}
We proceed by induction on $k$, with the cases $k=0$ and $k=1$ being immediate.
Now fix $k \ge 2$, and suppose $\dim \cV(n'; u^1, \cdots, u^\ell) < \infty$ when $\ell < k$.
First observe that when $m_1 < -n$ and $m_1 + \cdots m_k = -n$ we must have $m_2 + \cdots m_k > 0$, and hence
\[
v^1_{m_1} \cdots v_{m_k}^k \Omega = 0.
\]
Next observe that
\begin{align*}
&\cV(n; v^1, \cdots, v^k) = \\ & \quad\sum_{{m_1}=-n}^0 v^1_{m_1} \cV(n+{m_1}; v^2, \ldots, v^k) + \spann \{ v^1_{m_1} \cdots v_{m_k}^k \Omega \, | \, m_1 + \cdots + m_k = -n, \, m_1 > 0\}
\end{align*}
and each of the subspaces $v^1_{m_1} \cV(n+m_1; v^2, \ldots, v^k)$ is finite-dimensional by the inductive hypothesis, so it suffices to show that the final term is finite-dimensional as well.

By the Borcherds commutator formula \eqref{eqn: BCF}, if $m_1 > 0$ and $m_1 + \cdots m_k = -n$ we have
\[
v^1_{m_1} \cdots v_{m_k}^k \Omega = \sum_{j=2}^k v^2_{m_2} \cdots [v^1_{m_1}, v^j_{m_j}] \cdots v^k_{m_k}\Omega \in \sum_{j=2}^k \sum_{s=0}^{d_1+d_j-1} \cV(n; v^2, \ldots, v^1_{s-d_1+1} v^j, \ldots v^k).
\]
The subspace on the right-hand side is finite-dimensional by the inductive hypothesis and independent of $m_1, \ldots, m_k$, and thus we conclude that $\cV(n; v^1, \ldots, v^k)$ is finite-dimensional as well.
\end{proof}

If $f \in \bbC[z^{\pm 1}]$, then  $Y^0(v,f)$ maps $\cV$ into $\cV$.
Our next lemma gives an estimate for these maps $Y^0(v,f)$ in terms of the $N$-Sobolev norm of $f$.
Recall that for $N \in \bbR_{\ge 0}$, the $N$-Sobolev norm on $C^\infty(S^1)$ is given by
\begin{equation}\label{eqn: sobolev norm}
\norm{f}_N = \left(\sum_{n \in \bbZ} \big|\hat f(n)\big|^2(1+n^2)^N\right)^{1/2}.
\end{equation}
We denote by $H^N(S^1)$ the Hilbert space completion of $C^\infty(S^1)$ under this norm, which consists of $L^2$-functions with finite $N$-Sobolev norm.
The locally convex topology on $C^\infty(S^1)$ is induced by the norms $\norm{\cdot}_N$, and a linear map from $C^\infty(S^1)$ to a Banach space is continuous precisely when it is bounded with respect to some $N$-Sobolev norm.
We then have the following estimate, which is a simplification of \cite[Lem.\! 2.8]{RaymondTanimotoTener22}.
\begin{lem}\label{lem: finite energy estimate}
For all $v_1, \ldots, v_k, u \in \cV$, $u' \in \cV'$, and Laurent polynomials $f_1, \ldots, f_k \in \bbC[z^{\pm 1}]$, we have
\[
\abs{\bil{Y^0(v_k,f_k) \cdots Y^0(v_1,f_1)u,u'}_{\cV,\cV'}} \le C \norm{f_1}_N \cdots \norm{f_k}_N.
\]
The number $N$ depends only on the $v_j$, and the constant $C$ depends on the $\{v_j\}_{j=1,\cdots,k}$, $u$, and $u'$.
\end{lem}

\paragraph{The auxiliary domain $\cD$ and the topology on it.}
By Lemma~\ref{lem: va subspace finite dimensional}, the assignment $(f_1, \ldots, f_k) \mapsto Y^0(v_k,f_k) \cdots Y^0(v_1,f_1)u$ gives a map
\[
\bbC[z^{\pm 1}]^k \to \prod_{n=0}^\infty \cV(n; v_1, \ldots, v_k, u)
\]
with each space $\cV(n; v_1, \ldots, v_k, u)$ finite-dimensional.
By Lemma~\ref{lem: finite energy estimate}, this extends to a continuous multilinear map again taking values in $\prod_{n=0}^\infty \cV(n; v_1, \ldots, v_k, u) \subset \widehat \cV$.
Thus for each $v_1, \ldots, v_k, u \in \cV$, there exists a unique continuous multilinear map
\[
X_{v_1, \ldots, v_k,u} : C^\infty(S^1)^k \to \widehat \cV
\]
such that when $f_1, \ldots, f_k \in \bbC[z^{\pm 1}]$ we have
\begin{equation}\label{eqn: Xvv is product of fields}
X_{v_1, \ldots, v_k,u}(f_1, \ldots, f_k) = Y^0(v_k,f_k) \cdots Y^0(v_1,f_1)u.
\end{equation}
Let $\cD_0 = \bbC \Omega$, and for $k= 1, 2, \ldots$ set 
\[
\cD_k = \spann \{ X_{v_1, \ldots, v_k,\Omega}(f_1, \ldots, f_k) \, | \, v_j \in \cV, f_j \in C^\infty(S^1)\} \subset \widehat \cV.
\]
We have $\cD_k \subset \cD_{k+1}$ by considering $v_1=\Omega$.
Let $\cD = \bigcup_{k=0}^\infty \cD_k \subset \widehat \cV$, equipped with the subspace topology (i.e.\! the weak topology induced by the linear functionals $\cV'$, in which a sequence (or net) $\Phi_j \in \cD$ converges to $\Phi$ if and only if $\lambda(\Phi_j)$ converges to $\lambda(\Phi)$ for all $\lambda \in \cV'$).

\begin{lem}\label{lem: exist Yvf}
For all $v \in \cV$ and $f \in C^\infty(S^1)$ there exists a unique continuous linear map $Y(v,f):\cD\to\cD$ such that:
\begin{enumerate}[i)]
\item $Y(v,f)|_{\cV} = Y^0(v,f)$.
\item The expressions $Y(v_1,f_1) \cdots Y(v_k,f_k)\Omega$ are (jointly) continuous in the functions $f_j$.
\end{enumerate}
In addition, we have $\cD = \spann \{ Y(v_1, f_1) \cdots Y(v_k,f_k)\Omega \, | \, k \in \bbZ_{\ge 0}, v_j \in \cV, f_j \in C^\infty(S^1)\}$.
\end{lem}
\begin{proof}
We first consider uniqueness.
When $f_1, \ldots, f_k$ are Laurent polynomials, the condition $Y(v,f)|_{\cV} = Y^0(v,f)$ determines the value of $Y(v_1,f_1) \cdots Y(v_k,f_k)\Omega \in \cV$.
The value of such expressions in $\cD$ is then uniquely determined by continuity in the functions $f_j$.

We now show existence.
We wish to define $Y(v,f)$ on $X_{v_1, \ldots, v_k,\Omega}(f_1, \ldots, f_k) \in \cD_k$ by the formula
\begin{equation}\label{eqn: Yvf def}
Y(v,f)X_{v_1, \ldots, v_k,\Omega}(f_1, \ldots, f_k) = X_{v,v_1, \ldots, v_k,\Omega}(f,f_1, \ldots, f_k),
\end{equation}
but must check that this is well-defined.

First, consider if $f$ is a Laurent polynomial.
The modes $v_n$ map $\cV(m)$ to $\cV(m-n)$, and thus the adjoint (transpose) operator $v_n^*$ maps $\cV(m)^*$ into $\cV(m+n)^*$.
Hence there is an adjoint map $Y^0(v,f)^*: \cV' \to \cV'$ such that for $u \in \cV$ and $u' \in \cV'$ we have
\[
\bil{Y^0(v,f)u, u'}_{\cV,\cV'} = \bil{u,Y^0(v,f)^*u'}_{\cV, \cV'}.
\]
Thus if $f, f_1, \ldots, f_k$ are Laurent polynomials we have
\begin{align*}
\bil{X_{v,v_1, \ldots, v_k,u}(f,f_1, \ldots, f_k),u'} &= \bil{Y^0(v,f)Y^0(v_k,f_k) \cdots Y^0(v_1,f_1)u, u'}\\
& = \bil{Y^0(v_k,f_k) \cdots Y^0(v_1,f_1)u, Y^0(v,f)^*u'}\\
 &= \bil{X_{v_1, \ldots, v_k,u}(f_1, \ldots, f_k),Y^0(v,f)^* u'}.
\end{align*}
As the first and last terms are jointly continuous in $f_1, \ldots, f_k$ by \eqref{eqn: Xvv is product of fields}, we have
\begin{equation}\label{eqn: X adjoint relation}
\bil{X_{v,v_1, \ldots, v_k,u}(f,f_1, \ldots, f_k),u'}_{\widehat \cV, \cV'} = \bil{X_{v_1, \ldots, v_k,u}(f_1, \ldots, f_k),Y^0(v,f)^* u'}_{\widehat \cV, \cV'}
\end{equation}
whenever $f$ is a Laurent polynomial and $f_1,\ldots,f_{k}$ are smooth functions on $S^1$.

We now argue that \eqref{eqn: Yvf def} is well-defined for all $f,f_1,\ldots,f_k \in C^{\infty}(S^1)$.
Let $\tilde X_k :(\cV \otimes C^\infty(S^1))^{\otimes k} \to \cD$ be the linear map corresponding to the multilinear map $X_{v_1, \ldots, v_k, \Omega}(f_1, \ldots, f_k)$, so that $\cD_k$ is the range of $\tilde X_k$.
Let $\cT = \bigoplus_{k=0}^\infty (\cV \otimes C^\infty(S^1))^{\otimes k}$, and let $\tilde X:\cT \to \cD$ be the map given by $\tilde X_k$ on the $k$th direct summand of $\cT$.
We wish to show that if $\Xi \in \cT$ and $\tilde X(\Xi) = 0$, then $\tilde X (v \otimes f \otimes \Xi) = 0$ as well.

Fix $\Xi$ as above.
By \eqref{eqn: X adjoint relation}, if $f$ is a Laurent polynomial we have for all $u' \in \cV'$
\[
\bil{\tilde X(v \otimes f \otimes \Xi), u'} = \bil{\tilde X(\Xi), Y^0(v,f)^*u'} = 0,
\]
and so $\tilde X(v \otimes f \otimes \Xi) = 0$.
On the other hand, $\tilde X(v \otimes f \otimes \Xi)$ is continuous in $f$, and so $\tilde X(v \otimes f \otimes \Xi) = 0$ vanishes for all $f \in C^\infty(S^1)$.
Thus there is a well-defined map $Y(v,f): \cD \to \cD$ satisfying \eqref{eqn: Yvf def}.

By construction we have $Y(v_1, f_1) \cdots Y(v_k,f_k)\Omega = X_{v_1, \ldots, v_k,\Omega}(f_1, \ldots ,f_k)$.
It follows immediately that such expressions span $\cD$. Moreover, since $X_{v_1, \ldots, v_k,\Omega}$ is continuous in the functions $f_j$ we have shown that property $ii)$ of the lemma holds for the operator $Y(v,f)$.
For property $i)$, we note that $Y(v,f)|_{\cV}$ agrees with $Y^0(v,f)$ by \eqref{eqn: Xvv is product of fields} when $f$ is a Laurent polynomial, and thus for all $f$ by continuity.
\end{proof}

\begin{rem}
From Lemma~\ref{lem: finite energy estimate} and Lemma~\ref{lem: exist Yvf} we have shown a version of the uniformly bounded order property for the operator-valued distributions $Y(v,f)$, namely that for any $v_1, \ldots, v_k \in \cV$, there is a positive number $N$ such that for every $u \in \cV$ the map $(f_1, \ldots, f_k) \mapsto Y(v_1, f_1) \ldots Y(v_k,f_k)u$ extends to a continuous map $H^N(S^1)^k \to \widehat \cV$. 
\end{rem}

We have constructed a family of operator-valued distributions $Y(v,f)$ on $\cD$.
We next consider M\"obius covariance of these distributions, which will hold when $v$ is quasiprimary.
To this end we introduce 
\[
\cV_{\mathrm{QP}} = \spann \{v \in \cV \, | \, v \text{ is quasiprimary}\} = \ker L_1.
\]
\begin{rem}
Let $\operatorname{Vac}(\cV) = \operatorname{ker} L_{-1}$
be the vacuum subalgebra of $\cV$. 
Then $\cV$ is completely reducible as a $\Lie(\Mob)_\bbC$-module if and only if $\cV(0) = \operatorname{Vac}(\cV)$ and $\operatorname{Vac}(\cV) \cap \operatorname{im} L_1 = 0$. If $\cV$ is completely reducible as a $\Lie(\Mob)_\bbC$-module, then $\cV$ is spanned by vectors of the form $L_{-1}^{n}v$ for $n \geq 0$, where $v \in \cV_{\mathrm{QP}}$. It follows from the axioms that $L_{-1}^{n}v \in \mathbb{C} v_{(-n-1)}\Omega$, and thus $\cV$ is generated by quasiprimary vectors. More details and proofs are given in \cite[see Prop.\! 4.9]{Kac98}.
\end{rem}

Note that $\cV_{\mathrm{QP}} = \bigoplus_{n=0}^\infty \cV_{\mathrm{QP}} \cap \cV(n)$, so that every vector in $\cV_{\mathrm{QP}}$ may be written uniquely as a sum of homogeneous quasiprimary components.
Let
\[
\cD_{\mathrm{QP}} = \spann \{ Y(v_1,f_1) \cdots Y(v_k,f_k)\Omega \, | \, k \in \bbZ_{\ge 0}, v_j \in \cV_{\mathrm{QP}}, f_j \in C^\infty(S^1)\} \subseteq \cD.
\]
Let us assume that $\cV$ is generated by $\cV_{\mathrm{QP}}$ as a vertex algebra, in which case $\cV \subset \cD_{\mathrm{QP}}$.
If $\cV$ is completely reducible as a $\Lie(\Mob)_{\mathbb{C}}$-module (i.e.\! if it is spanned by vectors of the form $L_{-1}^k v$ with $v \in \cV_{\mathrm{QP}}$), then it is evidently generated by $\cV_{\mathrm{QP}}$ as a vertex algebra and moreover $\cD_{\mathrm{QP}} = \cD$ as $Y(L_{-1}v,f) = Y(v,\rmi f'- m f)$ when $v \in \cV(m)$.
\begin{rem}
      Note that the requirement that $\cV$ is generated by $\cV_{\mathrm{QP}}$ is strictly weaker than the requirement that $\cV$ is completley reducible as a $\Lie(\Mob)_{\mathbb{C}}$-module. An example is the $\beta \gamma$-ghost VOA  with $c = 2$ (see \cite{AllWood20} for a comprehensive discussion). This VOA is generated by two quasiprimary vectors, namely $\beta_{(-1)}\Omega$ and $\gamma_{(-1)}\Omega$, that have conformal dimension $1$ and $0$, respectively.
      In this case, $\cV(0) = \spann \{ \gamma^{n}_{(-1)}\Omega  \ | \ n \in \mathbb{Z}_{\geq 0}\}$ and each of these vectors is the highest-weight vector of a highest-weight $\Lie(\Mob)_{\mathbb{C}}$-module. 
      It is straightforward to show that the vectors $\{ L_{-1}\gamma^{n}_{(-1)}\Omega \ | \ n \in \mathbb{Z}_{>0}\}$ each generate a $\Lie(\Mob)_{\mathbb{C}}$-submodule inside the corresponding highest-weight module. Alternatively, one can also verify that $L_{1}\beta_{(-1)}\gamma_{(-1)}\Omega \in \mathbb{C}\Omega$ and $\beta_{(-1)}\gamma_{(-1)}\Omega \ne 0$, and so $\operatorname{Vac}(\cV) \cap \operatorname{im} L_1 \neq 0$.
\end{rem}

We will now construct a representation of $U:\Mob \to \cL(\cD_{\mathrm{QP}})$ for which the Wightman fields $Y(v,f)|_{\cD_{\mathrm{QP}}}$ are covariant for all $v \in \cV_{\mathrm{QP}}$ and for which the vacuum $\Omega$ is invariant.
Note that such a representation is unique if it exists, as the covariance condition implies 
\[
U(\gamma)Y(v_1,f_1) \cdots Y(v_k,f_k)\Omega = Y(v_1,\beta_{d_1}(\gamma)f_1) \cdots Y(v_k, \beta_{d_k}(\gamma)f_k)\Omega.
\]
So the difficulty is in showing that a linear map satisfying the above condition exists, i.e.\! showing that if a linear combination of vectors of the form 
\[
Y(v_1,f_1) \cdots Y(v_k,f_k)\Omega
\]
vanishes, then so does the corresponding linear combination of 
\[
Y(v_1,\beta_{d_1}(\gamma)f_1) \cdots Y(v_k, \beta_{d_k}(\gamma)f_k)\Omega.
\]

We first extend the representation of $\Lie(\Mob)_\bbC$ furnished by the M\"obius vertex algebra structure on $\cV$ to a representation on $\cD_{\mathrm{QP}}$.
Recall that $\Lie(\Mob)_\bbC$ is spanned by complexified vector fields $g(\rme^{\rmi \vartheta})\frac{d}{d\vartheta}$ on the circle, where $L_k$ corresponds to $-\rmi \rme^{\rmi k\vartheta}\frac{d}{d\vartheta}$.

Let $\cV^*$ be the algebraic dual of $\cV$, and note that the adjoint operators $L_k^*:\cV^* \to \cV^*$ leave $\cV'$ invariant.
We claim that the closure of the graph $\Gamma(L_k) \subset \cV \times \cV$ in $\widehat \cV \times \widehat \cV$ is the graph of a densely defined linear operator.
Indeed, suppose that $v_j$ is a net in $\cV$ such that $v_j \to 0$ and $L_kv_j \to v$ in $\widehat\cV$.
Then for any $\lambda \in \cV'$ it holds that
\[
\lambda(v) = \lim_j \lambda(L_k v_j) = \lim_j (L_k^*\lambda)(v_j) = 0.
\]
As $\cV'$ separates points in $\widehat \cV$, we conclude that $v=0$ and that the closure of $\Gamma(L_k)$ is the graph of a densely defined operator as claimed.
Taking linear combinations we obtain a densely-defined operator on $\widehat \cV$ for every $X \in \Lie(\Mob)$, which we denote by $\pi(X)$.

\begin{lem}\label{lem: LieMob rep extends to D}
Let $\cV$ be a M\"obius vertex algebra that is generated by a set of quasiprimary vectors.
Then for any $g\tfrac{d}{d\vartheta} \in \Lie(\Mob)$ the domain of $\pi(g\tfrac{d}{d\vartheta})$ contains $\cD_{\mathrm{QP}}$ and $\pi(g\tfrac{d}{d\vartheta})$ leaves $\cD_{\mathrm{QP}}$ invariant.
Moreover, if $v$ is quasiprimary with conformal dimension $d$ we have 
\[
[\pi(g\tfrac{d}{d\vartheta}),Y(v,f)] = Y(v,(d-1)\tfrac{dg}{d\vartheta}f - g \tfrac{df}{d\vartheta})
\]
as endomorphisms of $\cD_{\mathrm{QP}}$.
\end{lem}
\begin{proof}
If $v \in \cV_{\mathrm{QP}}(d)$ then the commutation relations between $Y(v,z)$ and the $L_k$ from the definition of a M\"obius vertex algebra imply that when $f \in \bbC[z^{\pm 1}]$ is a Laurent polynomial we have
\[
[\pi(g \tfrac{d}{d\vartheta}), Y(v,f)] = Y\big(v, (d-1)\tfrac{dg}{d\vartheta}f - g \tfrac{df}{d\vartheta}\big)
\]
as endomorphisms of $\cV$.
Thus if $f_1, \ldots, f_k$ are Laurent polynomials and $v_1, \ldots, v_k \in \cV_{\mathrm{QP}}$, then we have 
\begin{align}\label{eqn: pi(g) on smeared fields}
\pi(g\tfrac{d}{d\vartheta})Y(v_1,f_1) \cdots Y(v_k,f_k)\Omega = 
\sum_{j=1}^k Y(v_1,f_1) \cdots Y(v_j,(d_j-1)\tfrac{dg}{d\vartheta}f_j - g \tfrac{df_j}{d\vartheta}) \cdots Y(v_k,f_k)\Omega
\end{align}
where $d_j$ is the conformal dimension of $v_j$.
For arbitrary $f_1, \ldots, f_k \in C^\infty(S^1)$, choose sequences of Laurent polynomials\footnote{$f_{j,n} \in C^\infty(S^1)$
and it is not the $n$-th Fourier coefficient of $f_j$.} $f_{j,n}$ such that $\lim_{n\to\infty} f_{j,n} = f_j$ in $C^\infty(S^1)$, and observe that
\[
\lim_{n \to \infty} Y(v_1,f_{1,n}) \cdots Y(v_k,f_{k,n})\Omega = Y(v_1,f_1) \cdots Y(v_k,f_k)\Omega
\]
and
\begin{align*}
&\lim_{n \to \infty} Y(v_1,f_{1,n}) \cdots Y(v_j,(d_j-1)\tfrac{dg}{d\vartheta}f_{j,n} - g \tfrac{df_{j,n}}{d\vartheta}) \cdots Y(v_k,f_{k,n})\Omega
\\
=&\quad Y(v_1,f_1) \cdots Y(v_j,(d_j-1)\tfrac{dg}{d\vartheta}f_j - g \tfrac{df_j}{d\vartheta}) \cdots Y(v_k,f_k)\Omega
\end{align*}
in $\widehat \cV$ by Lemma~\ref{lem: exist Yvf}.
Hence $Y(v_1,f_1) \cdots Y(v_k,f_k)\Omega$ lies in the domain of $\pi(g \tfrac{d}{d\vartheta})$ and \eqref{eqn: pi(g) on smeared fields} holds for $f_j \in C^\infty(S^1)$.
It follows that $\pi(g \tfrac{d}{d\vartheta})$ leaves $\cD_{\mathrm{QP}}$ invariant and we have the desired commutation relation with smeared fields.
\end{proof}

We now turn to constructing the desired representation $U:\Mob \to \cL(\cD_{\mathrm{QP}})$.
The following lemma will allow us to define $U(\gamma)$ on a $\cV \subset \cD$.

\begin{lem}\label{lem: Laurent polynomial smearing vertex fields are Mobius covariant}
Let $\cV$ be a M\"obius vertex algebra which is generated by a set of quasiprimary vectors.
Then for any $\gamma \in \Mob$ there exists a unique linear map $U^0(\gamma):\cV \to \cD_{\mathrm{QP}}$ such that 
\[
U^0(\gamma)Y(v_1,f_1) \cdots Y(v_k,f_k)\Omega = Y(v_1,\beta_{d_1}(\gamma)f_1) \cdots Y(v_k, \beta_{d_k}(\gamma)f_k)\Omega.
\]
for all $v_1, \ldots, v_k \in \cV_{\mathrm{QP}}$ with conformal dimensions $d_j$ and all $f_j \in \bbC[z^{\pm 1}]$.
\end{lem}
\begin{proof}
Uniqueness is clear as the required formula for $U^0(\gamma)$ determines its value on $\cV$.
In order to show existence of $U^0(\gamma)$, we must show that if a linear combination of vectors of the form $Y(v^1,f_1) \cdots Y(v^k, f_k)\Omega$ vanishes, then so does the corresponding linear combination of vectors of the form $Y(v^1, \beta_{d_1}(\gamma)f_1) \cdots Y(v^k, \beta_{d_k}(\gamma)f_k)\Omega$.
We use standard ODE techniques.

The exponential map $\exp:\Lie(\Mob) \to \Mob$ is surjective, so we may choose $g\tfrac{d}{d\vartheta} \in \Lie(\Mob)$ such that $\exp(g \tfrac{d}{d\vartheta}) = \gamma$.
Let $\gamma_t = \exp(t g \tfrac{d}{d\vartheta})$ be the corresponding one-parameter subgroup of $\Mob$.
Let $v^1, \ldots, v^k \in \cV_{\mathrm{QP}}$ and consider the function $u:\bbR \to \cD_{\mathrm{QP}}$ given by
\[
u(t) = Y(v^1, \beta_{d_1}(\gamma_t)f_1) \cdots Y(v^k, \beta_{d_k}(\gamma_t)f_k)\Omega.
\]
We will now show that $u$ extends holomorphically to a neighborhood of $\bbR$ (when $\cD_{\mathrm{QP}}$ is given the weak topology induced by $\cV'$), and compute its derivative.

The map $\bbR \to \Mob \cong \mathrm{PSU}(1,1)$ given by $t \mapsto \gamma_t$ extends holomorphically to a neighborhood of $\bbR$ (taking values in \emph{complex} M\"obius transformations of the Riemann sphere $\cong \mathrm{PSL}_2(\bbC)$).
For each $t \in \bbR$, the M\"obius transformation $\gamma_t$ leaves $S^1$ invariant, and thus for a sufficiently small neighborhood of $\bbR$ the corresponding M\"obius transformations map $S^1$ into $\bbC^\times$.
Thus if $f \in \bbC[z^{\pm 1}]$ is a Laurent polynomial, the function $\bbR \times S^1 \to \bbC$ given by $(t,z) \mapsto (\beta_d(\gamma_t)f)(z)$ extends holomorphically to a neighborhood of $\bbR \times S^1$.
It follows that the map $\bbR \to C^\infty(S^1)$ sending $t \mapsto \beta_d(\gamma_t)f$ extends holomorphically to a neighborhood of $\bbR$.

Fix $\lambda \in \cV(n)^*$.
By Lemma~\ref{lem: exist Yvf}, the expressions $\lambda(Y(v^1,f_1) \cdots Y(v^k,f_k)\Omega)$ are jointly continuous in $f_j \in C^\infty(S^1)$.
Thus for fixed Laurent polynomials $f_1, \ldots, f_k \in \bbC[z^{\pm 1}]$, the function
\[
t \mapsto \lambda\big(Y(v^1,\beta_{d_1}(\gamma_t)f_1)\cdots Y(v^k,\beta_{d_k}(\gamma_t)f_k)\Omega\big)
\]
extends holomorphically to a neighborhood of $\bbR$.
As this neighborhood is independent of $\lambda$, the function 
\[
Y(v^1,\beta_{d_1}(\gamma_t)f_1)\cdots Y(v^k,\beta_{d_k}(\gamma_t)f_k)\Omega
\]
extends holomorphically to a neighborhood of $\bbR$, as previously claimed.

We now differentiate the above function of $t$.
A straightforward computation \cite[Eqn.\! (3.4)]{RaymondTanimotoTener22} shows that 
\begin{equation}\label{eqn: ddt betad in covariance proof}
\frac{d}{dt} \beta_d(\gamma_t)f = (d-1)\tfrac{dg}{d\vartheta}\beta_d(\gamma_t)f - g\tfrac{d}{d\vartheta}[\beta_d(\gamma_t)f]
\end{equation}
with the derivative taken in $C^\infty(S^1)$.
Comparing \eqref{eqn: ddt betad in covariance proof} with the commutation relation of Lemma~\ref{lem: LieMob rep extends to D} we obtain for any $\lambda \in \cV'$
\begin{align*}
\frac{d}{dt} &\lambda\big(Y(v^1,\beta_{d_1}(\gamma_t)f_1)\cdots Y(v^k,\beta_{d_k}(\gamma_t)f_k)\Omega\big)=\\ 
&= \sum_{j=1}^k \lambda\big(Y(v^1,\beta_{d_1}(\gamma_t)f_1)\cdots Y(v^j,\tfrac{d}{dt}\beta_{d_j}(\gamma_t)f_j) \cdots Y(v^k,\beta_{d_k}(\gamma_t)f_k)\Omega\big)\\
&= \sum_{j=1}^k \lambda\big(Y(v^1,\beta_{d_1}(\gamma_t)f_1)\cdots [\pi(g\tfrac{d}{d\vartheta}),Y(v^j,f_j)] \cdots Y(v^k,\beta_{d_k}(\gamma_t)f_k)\Omega\big)\\
&= \lambda\big(\pi(g\tfrac{d}{d\vartheta}) Y(v^1,\beta_{d_1}(\gamma_t)f_1)\cdots Y(v^k,\beta_{d_k}(\gamma_t)f_k)\Omega\big).
\end{align*}
Since the adjoint operator $\pi(g\tfrac{d}{d\vartheta})^*$ leaves $\cV'$ invariant, we may iterate the above argument to obtain
\begin{align*}
\frac{d^m}{dt^m} &\lambda\big(Y(v^1,\beta_{d_1}(\gamma_t)f_1)\cdots Y(v^k,\beta_{d_k}(\gamma_t)f_k)\Omega\big) = \\
&= \lambda\big(\pi(g\tfrac{d}{d\vartheta})^m Y(v^1,\beta_{d_1}(\gamma_t)f_1)\cdots Y(v^k,\beta_{d_k}(\gamma_t)f_k)\Omega\big).
\end{align*}
Since $\lambda$ was arbitrary we have
\begin{align}\label{eqn: repeated derivatives under Mobius transformation}
\frac{d^m}{dt^m} &Y(v^1,\beta_{d_1}(\gamma_t)f_1)\cdots Y(v^k,\beta_{d_k}(\gamma_t)f_k)\Omega = \\
&= \pi(g\tfrac{d}{d\vartheta})^m Y(v^1,\beta_{d_1}(\gamma_t)f_1)\cdots Y(v^k,\beta_{d_k}(\gamma_t)f_k)\Omega. \nonumber
\end{align}

We now complete the proof of existence of the map $U^0(\gamma):\cV \to \cD_{\mathrm{QP}}$.
Suppose that a certain linear combination of vectors of the form $Y(v^1,f_1)\cdots Y(v^k,f_k)\Omega$ vanishes.
That is, suppose we have
\[
\sum_{i=1}^\ell Y(v^{i,1}, f_{i,1}) \cdots Y(v^{i,k_i}, f_{i,k_i})\Omega = 0
\]
for $\ell,k_i \in \bbZ_{\ge 0}$, $v^{i,j}$ quasiprimary vectors with conformal dimension $d_{i,j}$, and $f_{i,j}$ Laurent polynomials.
Then, by the above,  the function
\begin{equation}\label{eqn: Mobius transformed vanished linear combo}
t \mapsto \sum_{i=1}^\ell Y(v^{i,1}, \beta_{d_{i,1}}(\gamma_t)f_{i,1}) \cdots Y(v^{i,k_i}, \beta_{d_{i,k_i}}(\gamma_t)f_{i,k_i})\Omega
\end{equation}
extends holomorphically to a neighborhood of $\bbR$, and by \eqref{eqn: repeated derivatives under Mobius transformation} the derivatives of \eqref{eqn: Mobius transformed vanished linear combo} at $t=0$ are given by
\begin{align*}
\frac{d^m}{dt^m} &\left. \sum_{i=1}^\ell Y(v^{i,1}, \beta_{d_{i,1}}(\gamma_t)f_{i,1}) \cdots Y(v^{i,k_i}, \beta_{d_{i,k_i}}(\gamma_t)f_{i,k_i})\Omega \, \right|_{t=0} =\\
&= \pi(g\tfrac{d}{d\vartheta})^m \sum_{i=1}^\ell Y(v^{i,1}, f_{i,1}) \cdots Y(v^{i,k_i}, f_{i,k_i})\Omega\\ 
&= 0.
\end{align*}
Since the Taylor series of \eqref{eqn: Mobius transformed vanished linear combo} at $t=0$ is identically zero, the function vanishes identically.
In particular, specializing to $t=1$ yields
\[
\sum_{i=1}^\ell Y(v^{i,1}, \beta_{d_{i,1}}(\gamma)f_{i,1}) \cdots Y(v^{i,k_i}, \beta_{d_{i,k_i}}(\gamma)f_{i,k_i})\Omega = 0.
\]
We conclude that the desired map $U^0(\gamma)$ is well-defined, as required.
\end{proof}

We now address the problem of extending $U^0(\gamma)$ to an endomorphism of $\cD_{\mathrm{QP}}$.

\begin{lem}\label{lem: vertex fields are Mobius covariant}
Let $\cV$ be a M\"obius vertex algebra that is generated by a set of quasiprimary vectors.
Then for any $\gamma \in \Mob$ there exists a unique linear map $U(\gamma) \in \cL(\cD_{\mathrm{QP}})$ such that 
\[
U(\gamma)Y(v_1,f_1) \cdots Y(v_k,f_k)\Omega = Y(v_1,\beta_{d_1}(\gamma)f_1) \cdots Y(v_k, \beta_{d_k}(\gamma)f_k)\Omega.
\]
for all $v_1, \ldots, v_k \in \cV_{\mathrm{QP}}$ with conformal dimensions $d_j$ and all $f_j \in C^\infty(S^1)$.
\end{lem}
\begin{proof}
Let $\Phi \in \cD_{\mathrm{QP}}$, and recall that $\cD_{\mathrm{QP}} \subset \widehat \cV = \prod_{n=0}^\infty \cV(n)$.
Thus we may canonically write $\Phi = \sum_{n=0}^\infty u^n$ with $u^n \in \cV(n)$ (and the sum converging in the weak topology induced by $\cV'$).
We would like to define $U(\gamma)\Phi = \sum_{n=0}^\infty U^0(\gamma)u^n$, but first must check convergence of the sum.
It suffices to consider a vector $\Phi = Y(v^1,f_1) \cdots Y(v^k,f_k)\Omega$ with $v_j \in \cV_{\mathrm{QP}}$ and $f_j \in C^\infty(S^1)$.
Recall from Lemma~\ref{lem: va subspace finite dimensional} that the continuous map multilinear map $C^\infty(S^1)^k \to \cD_{\mathrm{QP}}$ given by 
\[
(f_1, \ldots, f_k) \mapsto Y(v^1,f_1) \cdots Y(v^k,f_k)\Omega
\]
takes values in $\prod_{n=0}^\infty \cV(n; v^1, \ldots, v^k)$ with $\cV(n; v^1, \ldots, v^k)$ a finite-dimensional subspace of $\cV(n)$.
Thus by the universal property of the projective tensor product $\otimes_\pi$ \cite[Prop. 43.4]{Treves67} we have a continuous linear map
\[
C^\infty(S^1) \otimes_\pi \cdots \otimes_\pi C^\infty(S^1) \to \prod_{n=0}^\infty \cV(n; v^1, \ldots, v^k).
\]
As $\prod_{n=0}^\infty \cV(n; v^1, \ldots, v^k)$ is complete, by \cite[Thm. 5.2]{Treves67} (see also Appendix \ref{app: TVS and LCS} for the completion of topological vector spaces) we may extend this map to a continuous linear map 
\[
C^\infty(S^1) \hat \otimes_\pi \cdots \hat \otimes_\pi C^\infty(S^1) \to \prod_{n=0}^\infty \cV(n; v^1, \ldots, v^k)
\]
where $\hat \otimes_\pi$ is the completed projective tensor product (see \cite[\S 43]{Treves67}).
We have a natural isomorphism of topological vector spaces 
\[
C^\infty(S^1) \hat \otimes_\pi \cdots \hat \otimes_\pi C^\infty(S^1) \cong C^\infty((S^1)^k)
\]
by \cite[Thm. 56.1]{Treves67} (extended to the manifold $S^1$ via partition of unity).
Thus we conclude that there exists a continuous linear map
\[
S:C^\infty((S^1)^k) \to \prod_{n=0}^\infty \cV(n; v^1, \ldots, v^k)
\]
characterized by 
\[
S(f_1\otimes_\pi \cdots \otimes _\pi f_k) = Y(v^1,f_1) \cdots Y(v^k,f_k)\Omega.
\]

Now fix $v^1, \ldots, v^k \in \cV$ and $f_1, \ldots, f_k \in C^\infty(S^1)$, and consider
\[
Y(v^1,f_1) \cdots Y(v^k,f_k)\Omega = \Phi = \sum_{n=0}^\infty u^n,
\]
again with $u^n \in \cV(n)$.
Let $V(\gamma):C^\infty((S^1)^k) \to C^\infty((S^1)^k)$ be the continuous linear map such that
\[
V(\gamma)g_1(z_1) \cdots g_k(z_k) = (\beta_{d_1}(\gamma)g_1)(z_1) \cdots (\beta_{d_1}(\gamma)g_k)(z_k)
\]
for all $g_j \in C^\infty(S^1)$.
For $n \in \bbZ$, let $P_n: C^\infty((S^1)^k) \to C^\infty((S^1)^k)$ be the natural projection onto the closed span of monomials $z_1^{n_1} \cdots z_k^{n_k}$ with $n_1 + \cdots + n_k = -n$ (whose kernel is spanned by monomials with $n_1 + \cdots + n_k \ne -n$).
For $F(z_1, \ldots, z_k) = f_1(z_1) \cdots f_k(z_k)$,
by construction we have $S(P_n F) = u^n$ and $S(V(\gamma)P_n F) = U^0(\gamma)u^n$,
where $U^0$ is defined in Lemma \ref{lem: Laurent polynomial smearing vertex fields are Mobius covariant}.
Since $\sum_{n\in \bbZ} P_n = \operatorname{id}$ (with convergence pointwise as operators on $C^\infty((S^1)^k)$) and both $S$ and $V(\gamma)$ are continuous, we have convergence of the sum
\[
\sum_{n=0}^\infty U^0(\gamma)u^n = \sum_{n=0}^\infty S(V(\gamma)P_nF)
\]
to $S(V(\gamma)F) = Y(v_1,\beta_{d_1}(\gamma)f_1) \cdots Y(v_k, \beta_{d_k}(\gamma)f_k)\Omega$.

As the action of $U^0(\gamma)$ on $u^n \in \cV(n)$ is well-defined by Lemma \ref{lem: Laurent polynomial smearing vertex fields are Mobius covariant}
and does not depend on the choice of $\{v^j\}$,
we have obtained both a well-defined map $U(\gamma)$ given by $U(\gamma)\Phi = \sum_{n=0}^\infty U^0(\gamma) u^n$ along with the required covariance relation.
\end{proof}

Since $\beta_d(\gamma_1)\beta_d(\gamma_2) = \beta_d(\gamma_1 \circ \gamma_2)$, the maps $U(\gamma)$ furnish a representation of $\Mob$ on $\cD_{\mathrm{QP}}$. 
For $v \in \cV_{\mathrm{QP}}$ the operator-valued distribution $Y(v,f)|_{\cD_{\mathrm{QP}}}$ is evidently covariant with respect to this representation.

\begin{thm}\label{thm: VA to Wightman}
Let $\cV$ be a M\"obius vertex algebra, let $S \subset \cV$ be a set of quasiprimary vectors that generate $\cV$ as a vertex algebra, and let 
\[
\cD_S = \spann \{ Y(v_1, f_1) \cdots Y(v_k,f_k)\Omega \, | \, k \in \bbZ_{\ge 0}, v_j \in S, f_j \in C^\infty(S^1) \}.
\]
Let $\cF_S = \{ Y(v,f)|_{\cD_S} \, | \, v \in S, f \in C^\infty(S^1) \}$ and for $\gamma \in \Mob$ let $U_S(\gamma) = U(\gamma)|_S$.
Then $(\cF_S,\cD_S,U_S,\Omega)$ is a (not-necessarily-unitary) M\"obius-covariant Wightman CFT.
\end{thm}
\begin{proof}
We have a family of operator-valued distributions $\cF_S$ on $\cD_S \subset \cD_{\mathrm{QP}} \subset \widehat \cV$.
Note that $\cV \subset \cD_S$ since $S$ generates $\cV$.
By Lemma~\ref{lem: exist Yvf} we have $\cV' \subset \cD_{\cF,S}^*$, where we note that it suffices to check continuity of $\bil{Y(v_1,f_1) \cdots Y(v_k,f_k)\Phi,u'}$ in the special case $\Phi=\Omega$ since $\cD_S$ is generated from $\Omega$ by $\cF_S$.
Hence $\cD_{\cF,S}^*$ separates points, as $\cV'$ separates points in $\widehat \cV$, and so $\cF_{S}$ acts regularly.
The subspace $\cD_\cS$ is invariant under $U$ by Lemma~\ref{lem: vertex fields are Mobius covariant}, and by the same lemma the fields $Y(v,f)$ are M\"obius covariant, which verifies the first axiom of a Wightman CFT.

We now check the locality axiom.
Let $v_1,v_2 \in S$, let $u \in \cV$ and let $u' \in \cV'$.
By the vertex algebra locality axiom, the formal distribution $(z_1 - z_2)^N \bil{ [Y(v_1,z_1), Y(v_2,z_2)]u,u'}$ vanishes for $N$ sufficiently large, and thus the corresponding distribution $(f_1,f_2) \mapsto \bil{[Y(v_1,f_1),Y(v_2,f_2)]u,u'}$ is supported on the diagonal $z_1=z_2$ (see \cite[Cor. 2.2]{Kac98} and \cite[Prop. A.1]{CKLW18}).
Hence when $f_1$ and $f_2$ have disjoint support we have 
\begin{equation*}
[Y(v_1,f_1), Y(v_2,f_2)]u=0 \quad \mathrm{for \ all} \ u \in \cV.
\end{equation*}
That is,
\[
[Y(v_1,f_1),Y(v_2,f_2)]Y(a_1, g_1) \cdots Y(a_k,g_k)\Omega = 0
\]
for all $a_j \in S$ and $g_j \in \bbC[z^{\pm 1}]$.
By the joint continuity of such expressions in $g_j$ (which shows that $\cV$ is $\cF$-weakly dense in $\cD_S$) and the cyclicity of $\Omega$, we see that $[Y(v_1,f_1),Y(v_2,f_2)]\Phi = 0$ for all $\Phi \in \cD_S$, and thus the locality axiom holds.

The vacuum axiom holds by construction, and the spectrum condition holds by Lemma~\ref{lem: pos eigenvalues dense implies spectrum condition}.
\end{proof}

\subsection{From Wightman CFTs to vertex algebras}\label{sec: Wightman to VOA}

Let $\cF$ be a Wightman CFT with domain $\cD$, with vacuum vector $\Omega$ and representation $U:\Mob \to \cL(\cD)$.
Let $\cV(n) \subset \cD$ be the finite energy subspace
\[
\cV(n) = \spann \{ \varphi_1(e_{j_1}) \cdots \varphi_k(e_{j_k})\Omega \, | \, k \in \bbZ_{\ge 0}, \sum j_i = -n, \varphi_i \in \cF\},
\]
where $e_j(z)=z^j$, and let $\cV = \bigoplus_{n \ge 0} \cV(n) \subset \cD$.
Note that when $n<0$ we have $\cV(n)=0$ by the spectrum condition of a Wightman CFT, and $\cV$ is $\cF$-strongly dense in $\cD$.

We will show that $\cV$ carries the structure of a M\"obius vertex algebra generated by the point-like quasiprimary fields corresponding to $\cF$.
For $\varphi \in \cF$ with conformal dimension $d$, the corresponding point-like field is a formal sum
\[
\hat \varphi(z) = \sum_{n \in \bbZ} \varphi(e_n)z^{-n-d} .
\]
The key steps are to establish the vertex algebra locality condition
\begin{equation*}\label{eqn: todo vertex locality}
(z-w)^N[\hat \varphi(z), \hat \psi(w)] = 0
\end{equation*}
for $N$ sufficiently large, as well as differentiating the representation $U$ to a representation of $\Lie(\Mob)_\bbC$ for which we have the infinitesimal M\"obius covariance condition
\begin{equation*}\label{eqn: todo mobius covariance}
[L_m, \hat \varphi(z)] = \big(z^{m+1} \tfrac{d}{dz} + (m+1)z^m d\big) \hat \varphi(z), \quad m = -1,0,1.
\end{equation*}
From there, we will invoke general results that say that families of covariant local fields produce vertex algebras (see \cite[Thm. 4.5]{Kac98} for the case of vertex algebras, or more specifically \cite[Thm. A.1]{RaymondTanimotoTener22} for a slight variant for M\"obius vertex algebras).

We begin by establishing M\"obius covariance.

\begin{lem}\label{lem: Lie Mob covariance}
There is a unique representation $\pi:\Lie(\Mob)_\bbC \to\cL(\cD)$ such that for all $\varphi \in \cF$ with conformal dimension $d$ and all $g\tfrac{d}{d\vartheta} \in \Lie(\Mob)$ we have $\pi(g\tfrac{d}{d\vartheta})\Omega=0$ and 
\[
[\pi(g\tfrac{d}{d\vartheta}), \varphi(f)] = \varphi\big((d-1)\tfrac{dg}{d\vartheta} - g \tfrac{df}{d\vartheta}\big).
\]
\end{lem}
\begin{proof}
Uniqueness of such a representation follows immediately from the cyclicity of the vacuum (W\ref{itm: W Vacuum}).
Let $g \tfrac{d}{d\vartheta} \in \Lie(\Mob)$, and let $\gamma_t \in \Mob$ be the associated one-parameter group.
We have
\[
U(\gamma_t)\varphi_1(f_1) \cdots \varphi_n(f_n)\Omega = \varphi_1(\beta_{d_1}(\gamma_t)f_1) \cdots \varphi_n(\beta_{d_n}(\gamma_t)f_n)\Omega,
\]
where $d_i$ is the conformal dimension of $\varphi_i \in \cF$.
The derivative of $\beta_d(\gamma_t)$ is given (as in \cite[Eqn.\! (3.4)]{RaymondTanimotoTener22}) by 
\[
\left.\frac{d}{dt}\right|_{t=0} \beta_d(\gamma_t)f = (d-1)\tfrac{dg}{d\vartheta}f - g\tfrac{df}{d\vartheta},
\]
with the derivative taken in the space of smooth functions on $S^1$.

Give $\cD$ the $\cF$-strong topology.
Since expressions $\varphi_1(f_1) \cdots \varphi_n(f_n)\Omega$ are jointly continuous in the $f_j$, we have
\begin{equation}\label{eqn:dertiv Ugammat on vector}
\left.\frac{d}{dt}\right|_{t=0} U(\gamma_t)\varphi_1(f_1) \cdots \varphi_n(f_n)\Omega
=
\sum_{j=1}^n \varphi_1(f_1) \cdots \varphi_j\big((d_j-1)\tfrac{dg}{d\vartheta}f_j - g \tfrac{df_j}{d\vartheta}\big) \cdots \varphi_n(f_n)\Omega.
\end{equation}
In particular, for every $\Phi \in \cD$ the expression $U(\gamma_t)\Phi$ is differentiable at $t=0$, and we define $\pi(g\tfrac{d}{d\vartheta})\Phi = \left.\tfrac{d}{dt}\right|_{t=0} U(\gamma_t)\Phi$.
We have $\pi(g\tfrac{d}{d\vartheta})\Omega=0$ by the M\"obius invariance of the vacuum, and from \eqref{eqn:dertiv Ugammat on vector} we obtain the desired commutation relation for $[\pi(g \tfrac{d}{d\vartheta}), \varphi(f)]$.
A direct calculation shows that $\pi$ is a Lie algebra representation.
\end{proof}

Recalling that $L_m = \pi( -\rmi \rme^{\rmi m\vartheta} \tfrac{d}{d\vartheta})$ for $m=-1,0,1$, one can apply Lemma~\ref{lem: Lie Mob covariance} term-by-term to the modes of $\hat \varphi(z)$ to deduce the infinitesimal covariance relation
\begin{equation}\label{eqn: pointlike mobius covariance}
[L_m, \hat \varphi(z)] = \big(z^{m+1} \tfrac{d}{dz} + (m+1)z^m d\big) \hat \varphi(z).
\end{equation}

We now turn our attention to establishing the vertex algebra locality condition.
Recall that $\cV'$ denotes $\bigoplus_{n=0}^\infty \cV(n)^*$; that is, the space of linear functionals on $\cV$ that are supported on finitely many $\cV(n)$.
By abuse of notation we write $\cD_\cF^* \cap \cV'$ for the subspace of $\cD_\cF^*$ consisting of linear functionals $\lambda$ such that $\lambda|_\cV \in \cV'$, and similarly for $\cD_\cF^* \cap \cV(n)^*$.
By Lemma~\ref{lem: pos eigenvalues dense implies spectrum condition} $\cD_\cF^* \cap \cV(n)^*$ separates points in $\cV(n)$, and so $\cD_\cF^* \cap \cV'$ separates points in $\cV$.
The endomorphism $L_{-1}=\pi(-\rmi \rme^{-\rmi\vartheta}\tfrac{d}{d\vartheta})$ of \eqref{eqn: pointlike mobius covariance} and Lemma~\ref{lem: Lie Mob covariance} gives an endomorphism of $\cD$ by the lemma.
Moreover, the adjoint (transpose) operator $L_{-1}^*$ leaves $\cD_\cF^*$ invariant, becauase if $\lambda \in \cD_\cF^*$, then we have by the same lemma
\begin{align*}
(L_{-1}^*\lambda)&(\varphi_1(f_1) \cdots \varphi_k(f_k)\Omega) =\\
&= \lambda\left(\sum_{j=1}^k \varphi_1(f_1) \cdots \varphi_j\big(-(d_j-1)\rme^{-\rmi\vartheta} f_j + \rmi \rme^{-\rmi \vartheta} \tfrac{df_j}{d\vartheta}\big) \cdots \varphi_k(f_k)\Omega\right),
\end{align*}
which depends continuously on the $f_j$, so $L_{-1}^* \lambda \in \cD_\cF^*$.
Hence $L_{-1}^*$ leaves $\cD_\cF^* \cap \cV'$ invariant, mapping $\cD_\cF^* \cap \cV(n)^*$ into $\cD_\cF^* \cap \cV(n-1)^*$.

If $\lambda$ is a linear functional on a vector space $V$ and $A(z_1,\ldots,z_k)$ is a formal series with coefficients in $V$, then we write $\lambda(A(z_1,\ldots,z_k))$ for the corresponding formal series with coefficients in $\bbC$.
\begin{lem}\label{lem: order of pole in product of wightman fields}
Let $\varphi_1,\varphi_2 \in \cF$ with conformal dimensions $d_1$ and $d_2$, respectively.
Then for every $\lambda \in \cV'$ the formal series 
\[
(z_1-z_2)^{d_1+d_2}\lambda\big(\hat\varphi_1(z_1)\hat\varphi_2(z_2)\Omega\big)
\]
is a polynomial in $z_1$ and $z_2$ after expanding $(z_1-z_2)^{d_1+d_2}$ using the binomial theorem.
\end{lem}
\begin{proof}

We use standard vertex algebra arguments which go through in the present context. 
From the positivity of the energy and the $L_0$- and $L_{-1}$-commutation relations \eqref{eqn: pointlike mobius covariance}, we can deduce (as in the proof of \cite[Thm. 3.11]{RaymondTanimotoTener22}) that $\hat\varphi_2(z_2)\Omega$ has only non-negative powers of $z_2$, and if $u:=\hat\varphi_2(z_2)\Omega\mid_{z_2=0}$ is the constant term, then $u \in \cV(d_2)$.
The formal power series $\rme^{z_2L_{-1}}u$ and $\hat \varphi(z_2)\Omega$ both solve the initial value problem $\tfrac{d}{dz_2} F(z_2) = L_{-1} F(z_2)$ with $F(0) = u$.
This initial value problem has a unique solution in $\cV[[z_2]]$, and we conclude $\hat\varphi_2(z_2)\Omega = \rme^{z_2L_{-1}}u$ as formal series.

Similarly, we consider the formal series in $z_1^{\pm 1}$ and $z_2$ given by $\rme^{-z_2L_{-1}}\hat\varphi_1(z_1)\rme^{z_2L_{-1}}$.
It satisfies the initial value problem $\tfrac{d}{dz_2} F(z_1,z_2) = -[L_{-1},F(z_1,z_2)]$ with $F(z_1,0)=\hat\varphi_1(z_1)$.
Taking each coefficient of $z_1^m$ separately, it is straightforward to see that this initial value problem has a unique solution in $\End(\cV)[[z_1^{\pm 1},z_2]]$.
Let $\iota_{z_1,z_2} \hat\varphi_1(z_1-z_2)$ denote the series in $\End(\cV)[[z_1^{\pm 1}, z_2]]$ obtained by expanding each term $(z_1-z_2)^m$ as a binomial series with positive powers of $z_2$.
This series satisfies the same initial value problem, and so we have
\[
\rme^{-z_2L_{-1}}\hat\varphi_1(z_1)\rme^{z_2L_{-1}} = \iota_{z_1,z_2} \hat\varphi_1(z_1-z_2).
\]

Putting the two calculations together, we obtain an identity of formal series
\[
\hat\varphi_1(z_1)\hat\varphi_2(z_2)\Omega =  \rme^{z_2L_{-1}}\iota_{z_1,z_2} \hat\varphi_1(z_1-z_2)u.
\]
Hence
\[
\lambda\big(\hat\varphi_1(z_1)\hat\varphi_2(z_2)\Omega\big) = (\rme^{z_2L_{-1}^*}\lambda)\big(\iota_{z_1,z_2}\hat\varphi_1(z_1-z_2)u\big).
\]
As $L_{-1}^*$ maps $\cV(n)^*$ into $\cV(n-1)^*$, it acts nilpotently on $\lambda$ and the sum defining $\rme^{z_1 L_{-1}^*}\lambda$ is finite.

Consider a term of this sum, which is of the form $((L_{-1}^*)^m\lambda)(\hat\varphi_1(z_1-z_2)u)$.
It suffices to prove the lemma for $\lambda \in \cV(d)^*$ and then take linear combinations, in which case there is at most one non-zero term in the sum defining this expression.
That is, if we write $\hat\varphi_1(z) = \sum \hat\varphi_{1,n} z^{-n-d_1}$ then
\[
((L_{-1}^*)^m\lambda)(\hat\varphi_1(z_1-z_2)u) = (z_1-z_2)^{-d_1-d_2+d-m} \lambda(\hat\varphi_{1,d_2-d+m}u).
\]
Since this term is non-zero only when $m \le d$, we have that
\[
(z_1-z_2)^{d_1+d_2} \iota_{z_1,z_2} ((L_{-1}^*)^m\lambda)(\hat\varphi_1(z_1-z_2)u) = \iota_{z_1,z_2} C (z_1-z_2)^{d-m}
\]
for a constant $C$, which is a polynomial in $z_1$ and $z_2$.
We conclude that 
\[
(z_1-z_2)^{d_1+d_2} \iota_{z_1,z_2} (\rme^{z_2L_{-1}^*}\lambda)(\hat\varphi_1(z_1-z_2)u)
\]
is a polynomial in $z_1$ and $z_2$, and we are done.
\end{proof}

\begin{lem}\label{lem: wightman to va locality}
Let $\cF$ be a M\"obius-covariant Wightman CFT, and let $\varphi_1,\varphi_2 \in \cF$.
Then $\hat \varphi_1$ and $\hat\varphi_2$ are local in the sense of vertex algebras.
\end{lem}
\begin{proof}
Let $X:C^\infty(S^1)\times C^\infty(S^1) \to \End(\cD)$ be the operator-valued distribution corresponding to the formal series $(z_1-z_2)^{d_1+d_2}[\hat\varphi_1(z_1), \hat \varphi_2(z_2)]$ after expanding out the binomial $(z_1-z_2)^{d_1+d_2}$.
More precisely, we first define $X$ on pairs of functions $(e_n,e_m)$, where $e_n(z)=z^n$, by
\[
(z_1-z_2)^{d_1+d_2}[\hat\varphi_1(z_1), \hat \varphi_2(z_2)] = \sum_{n,m \in \bbZ} X(e_n,e_m)z_1^{-n-1}z_2^{-m-1},
\]
and these coefficients lie in $\End(\cV)$.
However expanding $(z_1-z_2)^{d_1+d_2}$ we see that $X$ is a (finite) linear combination of distributions of the form 
\begin{equation}\label{eqn: Lambda binomial expansion}
(f,g) \mapsto [\varphi_1(e_n \cdot f), \varphi_2(e_m \cdot g)],
\end{equation}
which extends to a genuine distribution $X:C^\infty(S^1)\times C^\infty(S^1) \to \End(\cD)$ as claimed.
Moreover, we see from this formula that $X(f,g) = 0$ when $f$ and $g$ have disjoint support, i.e.\! the support of $X$ is contained in the diagonal of $S^1 \times S^1$.

Let $\lambda \in \cD_\cF^* \cap \cV'$, and note that since $\lambda \in \cD_\cF^*$ the distribution
\[
(f,g) \mapsto \lambda(X(f,g)\Omega)
\]
is indeed continuous in $f$ and $g$.
Applying Lemma~\ref{lem: order of pole in product of wightman fields} twice, we see that this distribution, which corresponds to the formal series $(z_1-z_2)^{d_1+d_2}\lambda([\hat\varphi_1(z_1),\hat\varphi_2(z_2)]\Omega)$,
is given by integration against a trigonometric polynomial.
As noted above this distribution (and hence the corresponding polynomial) has support contained in the diagonal of $S^1 \times S^1$, and thus must be identically zero.
Since $\cD_\cF^* \cap \cV'$ separates points in $\cV$ by Lemma~\ref{lem: pos eigenvalues dense implies spectrum condition} we conclude that $X(e_n,e_m)\Omega=0$ for all $n,m \in \bbZ$.
As $X(f,g)\Omega$ is $\cF$-weakly continuous in $f$ and $g$, this implies that $X(f,g)\Omega=0$ for all $f,g \in C^\infty(S^1)$.

Recall that $X$ is a linear combination of distributions of the form \eqref{eqn: Lambda binomial expansion}.
Hence if $f$ and $g$ are supported in an open, non-dense interval $I$ of the circle, then the Reeh-Schlieder property (Corollary~\ref{cor: reeh schlieder}) implies that $X(f,g)=0$.
Now choose three intervals that cover $S^1$ such that the union of any two is contained inside some interval, and let $\{\chi_i\}$ be a partition of unity subordinate to this cover.
Then $X(f,g) = \sum_{i,j=1}^3 X(f\chi_i,g\chi_j) = 0$ for arbitrary $f,g \in C^\infty(S^1)$.
In particular $X(e_n,e_m)=0$ for all $n,m \in \bbZ$, and we conclude that the formal series $(z_1-z_2)^{d_1+d_2}[\hat\varphi_1(z_1),\hat\varphi_2(z_2)]$ is identically zero, as desired.
\end{proof}

We can now state and prove one of our main results, constructing a M\"obius vertex algebra from a Wightman theory.

\begin{thm}\label{thm: Wightman to VA}
Let $\cF$ be a (not-necessarily-unitary) M\"obius-covariant Wightman CFT on $S^1$ with domain $\cD$, and let $\cV \subset \cD$ be given by
\[
\cV = \spann \{ \varphi_1(e_{j_1}) \cdots \varphi_k(e_{j_k})\Omega \, | \, k \in \bbZ_{\ge 0}, \, j_i \in \bbZ, \, \varphi_i \in \cF\}.
\]
Then there is a unique structure of M\"obius vertex algebra on $\cV$ such that for every $\varphi \in \cF$ with conformal dimension $d$ there is a quasiprimary $v_\varphi \in \cV(d)$ such that $\hat \varphi(z) = Y(v_\varphi,z) \in \End(\cV)[[z^{\pm 1}]]$.
The set $S=\{v_\varphi \, | \, \varphi \in \cF\}$ generates $\cV$.
\end{thm}
\begin{proof}
We equip $\cV$ with the representation of $\Lie(\Mob)_\bbC$ from Lemma~\ref{lem: Lie Mob covariance}.
To show that the point-like fields $\hat \varphi(z)$ generate a M\"obius vertex algebra, we invoke \cite[Thm. A.1]{RaymondTanimotoTener22} (see also \cite[Thm. 4.5]{Kac98}).
To invoke this theorem, we need to verify that:
\begin{enumerate}
\item $\cV = \bigoplus_{n \ge 0} \ker(L_0-n)$
\item $\Omega$ is $\Lie(\Mob)$-invariant
\item For every $\varphi \in \cF$, $\hat \varphi(z)\Omega$ has a removable singularity at $z=0$
\item For every $\varphi \in \cF$, there exists a $d_\varphi \in \bbZ_{\ge 0}$ such that
\[
[L_m, \hat \varphi(z)] = (z^{m+1} \tfrac{d}{dz} + (m+1)z^m d_\varphi) \hat \varphi(z) \qquad m=-1,0,1
\]
\item For every $\varphi,\psi \in \cF$, we have $(z-w)^N [\hat \varphi(z), \hat \psi(w)] = 0$ for $N$ sufficiently large
\item $\cV = \spann \{ \varphi_1(e_{j_1}) \cdots \varphi_k(e_{j_k})\Omega \, | \, k \ge 0, j_i \in \bbZ, \varphi_i \in \cF\}$.
\end{enumerate}
The first point follows from the fact that $\varphi_1(e_{j_1}) \cdots \varphi_k(e_{j_k})\Omega$ is an eigenvector for $L_0$ with eigenvalue $-\sum j_i$ by the commutation relation of Lemma~\ref{lem: Lie Mob covariance}.
The second point and fourth point also follow from Lemma~\ref{lem: Lie Mob covariance} along with Equation~\eqref{eqn: pointlike mobius covariance}.
The fifth point holds by Lemma~\ref{lem: wightman to va locality}, and the sixth point is the definition of $\cV$.

We now argue the third point, that $\hat \varphi(z)$ has a removable singularity at $z=0$.
The argument is the same as in \cite[Thm. 3.11]{RaymondTanimotoTener22}.
Let $\varphi_n = \varphi(e_n)$.\footnote{Note that $\varphi_n$ is the $n$-th mode of a single field $\varphi$
and not the $n$-th field. We use this notation only here and in the next paragraph.}

We must show that $\varphi_{-n}\Omega = 0$ for $n \le d-1$.
When $n < 0$ this identity holds by the spectrum condition which implies that $\ker(L_0-n) = 0$ for these $n$.
So we now consider $n=0, \ldots, d-1$.
From the $L_{-1}$-commutation relation of $\hat \varphi$ we have
\[
\varphi_{-n}\Omega = \tfrac{1}{n-d} L_{-1}\varphi_{-n+1}\Omega.
\]
We repeatedly apply this identity, starting with $n=0$, to obtain $0=\varphi_{0}\Omega = \cdots =\varphi_{-d+1}\Omega$, as desired.

Thus by \cite[Thm. A.1]{RaymondTanimotoTener22} there exists a unique structure of M\"obius vertex algebra on $\cV$, with the same $L_n$, such that for every $\varphi \in \cF$ with conformal dimension $d$ we have $Y(\varphi_{-d}\Omega,z) = \hat \varphi(z)$.
The vector $\varphi_{-d}\Omega$ is quasiprimary, as
\[
L_1 \varphi_{-d}\Omega = [L_1, \varphi_{-d}]\Omega = \lim_{z \to 0} [L_1, Y(\varphi_{-d}\Omega, z)]\Omega 
= \lim_{z \to 0} (z^2 \tfrac{d}{dz} + 2zd)Y(\varphi_{-d}\Omega, z)\Omega = 0.
\]
By the sixth point, the set $S$ in the statement of the theorem generates $\cV$.
This completes the proof of the existence statement.

For uniqueness, note that the set $\{\lim_{z \to 0} \hat \varphi(z)\Omega\}$ generates any vertex algebra satisfying the statement of the theorem.
The modes of the corresponding fields are determined by the fields $\varphi(z)$, and the modes of the remaining fields are then determined by the Borcherds product formula \eqref{eqn: BPF}.
The grading operator $L_0$ is determined by the requirement that the conformal dimension of $\lim_{z \to 0} \hat \varphi(z)\Omega$ matches the conformal dimension of $\varphi$.
The operators $L_{\pm 1}$ are then determined by the commutation relations with the generating fields.
We conclude that the M\"obius vertex algebra constructed above is the unique such structure satisfying the requirements of the theorem.
\end{proof}

As a corollary of the proof of Theorem~\ref{thm: Wightman to VA} we have that if $\varphi \in \cF$ is non-zero and has conformal dimension $d$ then
\begin{equation}\label{eqn: conf dim formula}
d = \inf \{n \in \bbZ_{\ge 0} \, | \, \varphi(e_{-n})\Omega \ne 0\},
\end{equation}
and in particular this gives a proof that the conformal dimension of a Wightman field is uniquely determined.

We conclude this section with a canonical realization of the domain $\cD$ of a Wightman CFT.

\begin{prop}\label{prop: canonical embedding of D}
Let $\cF$ be a M\"obius-covariant Wightman CFT on $S^1$ with domain $\cD$, and let $\cV \subset \cD$ be the corresponding M\"obius vertex algebra from Theorem~\ref{thm: Wightman to VA}.
Equip $\cD$ with the $\cF$-strong topology and equip $\widehat \cV = \prod_{n=0}^\infty \cV(n)$ with the weak topology induced by $\cV'$.
Then the identity map $\operatorname{id}_\cV$ extends to a (necessarily unique) injective continuous linear map $\iota:\cD \to \widehat \cV$.
\end{prop}
\begin{proof}
First, we claim that for any $\Phi \in \cD$ there exists a unique sequence $\Phi_n \in \cV(n)$ such that $\Phi = \sum \Phi_n$, converging in the $\cF$-strong topology.
We first consider existence.
It suffices to establish existence for $\Phi=\varphi_1(f_1) \cdots \varphi_k(f_k)\Omega$.
Arguing as in the proof of Lemma~\ref{lem: vertex fields are Mobius covariant}, there exists a continuous map $S:C^\infty((S^1)^k) \to \widehat \cD$ such that
\[
S(f_1\otimes_\pi \cdots \otimes_\pi f_k) = \varphi_1(f_1) \cdots \varphi_k(f_k)\Omega
\]
for all $f_j \in C^\infty(S^1)$, where $\widehat \cD$ is the completion of $\cD$ in the $\cF$-strong topology (see Appendix \ref{app: TVS and LCS}) and 
\[
(f_1 \otimes_\pi \cdots \otimes_\pi f_k)(z_1, \ldots, z_k) = f_1(z_1) \cdots f_k(z_k).
\]
Let $P_n:C^\infty((S^1)^k) \to C^\infty((S^1)^k)$ be the projection onto the closed span of monomials $z_1^{n_1} \cdots z_k^{n_k}$ with $n_1 + \cdots +n_k = -n$ (whose kernel is spanned by monomials with $n_1 + \cdots + n_k \ne -n$).
When $f_1, \ldots, f_k \in \bbC[z^{\pm 1}]$ we have 
\[
S(P_n(f_1\otimes_\pi \cdots \otimes_\pi f_k)) \in \cV(n; v_1, \ldots, v_k)
\]
where $v_j \in \cV$ is the vector corresponding to $\varphi_j$.
Since $\cV(n; v_1, \ldots, v_k)$ has finite dimension, and is therefore complete \cite[Thm. 4.10.3]{NariciBeckenstein11}, the composed map $SP_n$ takes values in $\cV(n; v_1, \ldots, v_k)$, and in particular in $\cV(n)$.
Thus if $\Phi = S(f_1\otimes_\pi \cdots \otimes_\pi f_k)$ and we set $\Phi_n = SP_n(f_1\otimes_\pi \cdots \otimes_\pi f_k)$,
then $\Phi = \sum \Phi_n$ in (the natural extension of) the $\cF$-strong topology on $\widehat \cD$, because
$\sum_n P_n(f_1\otimes_\pi \cdots \otimes_\pi f_k) = f_1\otimes_\pi \cdots \otimes_\pi f_k$ in $C^\infty((S^1)^k)$ and $S$ is continuous.

We now consider uniqueness of the sequence $\Phi_n$.
Suppose that $\sum \Phi_n = 0$ with $\Phi_n \in \cV(n)$ and the sum converging $\cF$-strongly.
Then any $\lambda \in \cV(m)^* \cap \cD_\cF^*$ extends to $\widehat \cD$ by continuity (see Appendix \ref{app: TVS and LCS}) and we have 
\[
0 = \lambda(\Phi) = \sum \lambda(\Phi_n) = \lambda(\Phi_m).
\]
As $\cV(m)^* \cap \cD_\cF^*$ separates points in $\cV(m)$ by Lemma~\ref{lem: pos eigenvalues dense implies spectrum condition} we see $\Phi_m = 0$, and since $m$ was arbitrary this establishes the uniqueness portion of the claim.

We now define $\iota:\cD \to \widehat \cV$ by $\iota(\Phi) = (\Phi_n)_{n \ge 0}$, where $\Phi_n \in \cV(n)$ is the unique sequence such that $\sum \Phi_n = \Phi$ with $\cF$-strong convergence.
This map is well-defined by the above claim and, by inspection, $\iota$ is injective and restricts to the identity on $\cV$.
It remains to check that $\iota$ is continuous from the $\cF$-strong topology to the weak topology on $\widehat \cV$ induced by $\cV'$.
By the universal property of the $\cF$-strong topology, it suffices to check that $\lambda(\iota \varphi_1(f_1) \cdots \varphi_k(f_k)\Omega)$ depends continuously on the $f_j$ for any $\lambda \in \cV(n)^*$.
By the calculation above we have
\begin{equation}\label{eqn: lambda iota}
\lambda(\iota \varphi_1(f_1) \cdots \varphi_k(f_k)\Omega) =
\lambda(SP_n(f_1 \otimes_\pi \cdots \otimes_\pi f_k)).
\end{equation}
We have seen that $SP_n$ is a continuous map with values in the finite-dimensional space $\cV(n; v_1, \ldots, v_k)$, and $\lambda|_{\cV(n; v_1, \ldots, v_k)}$ is evidently continuous.
We conclude that \eqref{eqn: lambda iota} is continuous in the $f_j$, and so $\iota$ is continuous as claimed.
\end{proof}

\subsection{Equivalence of categories}\label{sec: equivalence of categories}

We have constructions in Theorem~\ref{thm: VA to Wightman} and Theorem~\ref{thm: Wightman to VA} that produce Wightman CFTs from vertex algebras and vice versa.
In this section we show that these constructions are inverse to each other, or more precisely we show that they induce an equivalence of categories.
We now introduce the relevant categories.

A homomorphism $g:\cV \to \tilde \cV$ of M\"obius vertex algebras is a linear map that intertwines the representations of $\Lie(\Mob)$, maps the vacuum vector to the vacuum vector, and intertwines the modes:
\[
g(v_{(n)}u) = g(v)_{(n)}g(u) \qquad u,v \in \cV.
\]
Now suppose that $\cV$ have $\tilde \cV$ are equipped with choices of generating sets of quasiprimary vectors $S$ and $\tilde S$, respectively.
We say that $g$ is a morphism $(\cV,S) \to (\tilde \cV, \tilde S)$ if $g$ is a homomorphism of M\"obius vertex algebras and $g(S) \subset \tilde S$.
We write $\mathsf{MVA}^+$ for the category of M\"obius vertex algebras equipped with a choice of generating set of quasiprimary vectors, where the superscript indicates this choice of a generating set.

If $(\cF,\cD,U,\Omega)$ and $(\tilde \cF, \tilde \cD, \tilde U, \tilde \Omega)$ are M\"obius-covariant Wightman CFTs on $S^1$, then a morphism $\cF \to \tilde \cF$ is a linear map $g:\cD \to \tilde \cD$ and a function $g_*:\cF \to \tilde \cF$ such that $g(\Omega) = \tilde \Omega$, $g$ intertwines $U$ and $\tilde U$, and  $g \varphi(f) = (g_* \varphi)(f) g$ for all $\varphi \in \cF$ and $f \in C^\infty(S^1)$.
Note that $g_*$ is uniquely determined by $g$.
A straightforward calculation shows that a homomorphism $g$ is continuous when $\cD$ and $\tilde \cD$ are respectively given the $\cF$-strong and $\tilde \cF$-strong topologies, and similarly for the $\cF$-weak and $\tilde \cF$-weak topologies.
We write $\mathsf{MW}$ for the category of M\"obius-covariant Wightman CFTs on $S^1$.

\begin{lem}\label{lem: functoriality Wightman to VA}
Let $(\cF,\cD,U,\Omega)$ and $(\tilde \cF, \tilde \cD, \tilde U, \tilde \Omega)$ be a pair of M\"obius-covariant Wightman CFTs and let $(g,g_*)$ be a morphism $\cF \to \tilde \cF$.
Let $\cV \subset \cD$ and $\tilde \cV \subset \tilde \cD$ be the M\"obius vertex algebras constructed in Theorem~\ref{thm: Wightman to VA}, and let $S$ and $\tilde S$ be the respective sets of generating vectors.
Then $g(\cV) \subset \tilde \cV$ and $g|_\cV:(\cV,S) \to (\tilde \cV,\tilde S)$ is a morphism in $\mathsf{MVA}^+$.
\end{lem}
\begin{proof}

By definition $\cV$ is spanned by vectors of the form $\varphi_1(e_{j_1}) \cdots \varphi_k(e_{j_k})\Omega$ where $\varphi_i \in \cF$ and $e_j(z) = z^j$.
Since $(g,g_*)$ is a morphism we have
\[
g\varphi_1(e_{j_1}) \cdots \varphi(e_{j_k})\Omega = (g_*\varphi_1)(e_{j_1}) \cdots (g_*\varphi_k)(e_{j_k})\tilde \Omega \in \tilde \cV,
\]
so $g(\cV) \subset \tilde \cV$.

We next check that $g$ intertwines the representations of $\Lie(\Mob)$.
Let $h \tfrac{d}{d\vartheta} \in \Lie(\Mob)$ and let $\gamma_t \in \Mob$ be the corresponding one-parameter group.
We saw in the proof of Lemma~\ref{lem: Lie Mob covariance} that the representations of $\Lie(\Mob)$ on $\cD$ and $\tilde \cD$ are given by differentiating $U(\gamma_t)$, and so we have
\[
g \pi(h \tfrac{d}{d\vartheta})v = g \left.\frac{d}{dt}\right|_{t=0}U(\gamma_t)v
= \left.\frac{d}{dt}\right|_{t=0}\tilde U(\gamma_t)gv
= \tilde \pi(h \tfrac{d}{d\vartheta}) gv
\]
where we used that the derivatives are taken in the $\cF$- and $\tilde \cF$-weak topologies, and $g$ is continuous with respect to these topologies.

Now fix $\varphi \in \cF$ with conformal dimension $d$.
Let $d'$ be the conformal dimension of $g_*\varphi$, and we begin by arguing $d=d'$ provided $g_*\varphi \not \equiv 0$.
By \eqref{eqn: conf dim formula} we have
\[
d = \inf \{ n \in \bbZ_{\ge 0} \, | \, \varphi(e_{-n})\Omega \ne 0\}, \qquad d' = \inf \{ n \in \bbZ_{\ge 0} \, | \, (g_*\varphi)(e_{-n})\Omega \ne 0\}.
\]
As $g_*\varphi(e_{-n})\Omega = g \varphi(e_{-n})\Omega$ we have $d \le d'$, and we must show that $g \varphi(e_{-d})\Omega \ne 0$.
From the previous step we know that $gL_n = \tilde L_n g$, where as usual $L_n = \pi(-\rmi \rme^{\rmi n \vartheta} \tfrac{d}{d\vartheta})$ and similarly for $\tilde L_n$.
Thus for $n \ne d$ we have
\[
(g_*\varphi)(e_{-n})\tilde \Omega = g \varphi(e_{-n})\Omega = \tfrac{1}{n-d}g L_{-1}\varphi(e_{-n+1})\Omega = \tfrac{1}{n-d}\tilde L_{-1} (g_*\varphi)(e_{-n+1})\tilde \Omega.
\]
If $(g_*\varphi)(e_{-d})\tilde \Omega = 0$, we may apply the above relation repeatedly to $n=d+1, d+2, \ldots$ to conclude that $(g_*\varphi)(e_{-n}) \tilde \Omega = 0$ for all $n \in \bbZ_{\ge 0}$.
But then we would have $(g_*\varphi)(e_{-d'})\tilde \Omega = 0$, a contradiction.
We conclude that $d'=d$, which is to say that $\varphi$ and $g_*\varphi$ have the same conformal dimension provided $g_*\varphi \not \equiv 0$.

Next observe that $Y(gv,z) = (\widehat{g_*\varphi})(z)$, or equivalently $(gv)_{(n)} = (g_*\varphi)(e_{n-d+1})$.
We therefore have 
\[
g v_{(n)} = g \varphi(e_{n-d+1}) = (g_*\varphi)(e_{n-d+1})g = (gv)_{(n)}g.
\]
This means that $g$ intertwines the actions of modes of vectors $v$ corresponding to $\varphi \in \cF$, and since such vectors generate $\cV$ we can conclude that $g$ intertwines the actions of modes $v_{(n)}$ for all $v \in \cV$.
Moreover the identity $g \varphi(e_{-d})\Omega = (g_* \varphi)(e_{-d})\tilde \Omega$ implies that $gS \subset \tilde S$.
\end{proof}

\begin{lem}\label{lem: functoriality VA to Wightman}
Let $\cV$ and $\tilde \cV$ be M\"obius vertex algebras with generating sets $S$ and $\tilde S$, respectively.
Let $g:(\cV,S) \to (\tilde \cV, \tilde S)$ be a morphism in $\mathsf{MVA}^+$.
Let $(\cD,\cF,U,\Omega)$ and $(\tilde \cD, \cF,\tilde U, \tilde \Omega)$ be the M\"obius-covariant Wightman CFTs constructed in Theorem~\ref{thm: VA to Wightman}.
Then there is a unique morphism $(h,h_*): \cF \to \tilde \cF$ such that $h|_\cV = g$.
\end{lem}
\begin{proof}
For $v \in S$, we write $\varphi_v := Y(v,\cdot)$ for the corresponding Wightman field in $\cF$, and similarly for $\tilde v \in \tilde S$ we write $\tilde \varphi_{\tilde v}:=\tilde Y(\tilde v, \cdot)$ for the Wightman field in $\tilde \cF$.
For $v \in S$, we define $h_* \varphi_v = \tilde \varphi_{gv} \in \tilde \cF$.
Since $g$ is a morphism of vertex algebras we have for all $\varphi_1, \ldots, \varphi_k \in \cF$ and all $f_1, \ldots, f_k \in \bbC[z^{\pm 1}]$
\[
g \varphi_1(f_1) \cdots \varphi_k(f_k)\Omega = (h_*\varphi_1)(f_1) \cdots (h_*\varphi_k)(f_k)\tilde \Omega.
\]
Since morphisms of Wightman CFTs are continuous for the $\cF$- and $\tilde \cF$-weak topologies, we can see from the above formula that a morphism $(h,h_*)$ as in the statement of the lemma is necessarily unique.

Since $g$ intertwines the actions of $L_0$ and $\tilde L_0$, the adjoint operator $g^*$ maps $\tilde \cV'$ into $\cV'$.
As $\cV' \subset \cD_\cF^*$ and $\tilde \cV' \subset \tilde \cD_{\tilde \cF}^*$ (by Lemma~\ref{lem: exist Yvf}), and $\tilde \cV'$ separates points in $\tilde \cD$, it follows that the closure of the graph of $g:\cV \to \tilde \cV$  in $\cD \times \tilde \cD$ is again the graph of a densely defined linear map $h:\cD \to \tilde \cD$.
If $f_1, \ldots, f_k \in C^\infty(S^1)$, we may approximate each $f_j$ by Laurent polynomials $f_{j,n}$ to obtain
\[
h \varphi_1(f_{1,n}) \cdots \varphi_k(f_{k,n})\Omega = (h_*\varphi_1)(f_{1,n}) \cdots (h_*\varphi_k)(f_{k,n})\Omega \to (h_*\varphi_1)(f_{1}) \cdots (h_*\varphi_k)(f_{k})\Omega.
\]
Hence $h$ is defined on all of $\cD$ and $h\varphi(f) = (h_*\varphi)(f)h$ for all $f \in C^\infty(S^1)$ and $\varphi \in \cF$.
It follows immediately that $h$ also intertwines the representations $U$ and $\tilde U$, and we have shown that $(h,h_*)$ is a morphism $\cF \to \tilde \cF$.
\end{proof}

Lemmas~\ref{lem: functoriality Wightman to VA} and \ref{lem: functoriality VA to Wightman} upgrade the constructions of Theorem~\ref{thm: Wightman to VA} and \ref{thm: VA to Wightman} to a pair of functors $F:\mathsf{MW} \to \mathsf{MVA}^+$ and $G:\mathsf{MVA}^+ \to \mathsf{MW}$.
In showing that these are an equivalence of categories, it will be helpful to note that if $\cF$ is a Wightman CFT with domain $\cD$, then the vertex algebra $\cV:=F(\cF)$ is a subspace $\cV \subset \cD$.
Conversely, if $\cV \in \mathsf{MVA}^+$ and $\cD$ is the domain of the Wightman CFT $G(\cV)$, then $\cV \subset \cD$.

\begin{lem}\label{lem: equivalence isos}
We have the following.
\begin{enumerate}[i)]
\item Let $(\cF,\cD,U,\Omega)$ be a M\"obius-covariant Wightman CFT, let $\cV = F(\cF)$ with $\cV \subset \cD$.
Let $(\tilde \cF, \tilde \cD, \tilde U, \tilde \Omega) = G(\cV)$ with $\cV \subset \tilde \cD$.
Then there is a unique isomorphism $(g,g_*): \cF \to \tilde \cF$ such that $g|_\cV = \mathrm{id}$.
\item Let $\cV$ be a M\"obius vertex algebra, let $(\cF,\cD,U,\Omega) = G(\cV)$ be the corresponding M\"obius-covariant Wightman CFT with $\cV \subset \cD$.
Let $\tilde \cV \subset \cD$ be the M\"obius vertex algebra $F(\cF)$.
Then $\tilde \cV = \cV$ as M\"obius vertex algebras.
\end{enumerate}
\end{lem}
\begin{proof}
We first consider (i).
Uniqueness of such an isomorphism follows from the fact that an isomorphism $g:\cD \to \tilde \cD$ is $\cF$-strong continuous and $\cV \subset \cD$ is $\cF$-strong dense.
We now consider existence.
By construction there is a canonical bijection $\cF \to \tilde \cF$ which we denote by $\varphi \mapsto \tilde \varphi$.
We must verify that there exists a corresponding bijection $\cD \to \tilde \cD$.
We have $\cV \subset \cD$ and $\cV \subset \tilde \cD$.
By construction we have $\tilde \cD \subset \widehat \cV$, and Proposition~\ref{prop: canonical embedding of D} provides a map $\iota:\cD \hookrightarrow \widehat \cV$.
We have $\iota \varphi_1(f_1) \cdots \varphi_k(f_k)\Omega = \tilde \varphi_1(f_1) \cdots \tilde \varphi_k(f_k)\Omega$ when $f_j \in \bbC[z^{\pm 1}]$, and since both sides are continuous in the functions $f_j$ this extends to all $f_j \in C^\infty(S^1)$.
We conclude that $\iota$ maps $\cD$ into $\tilde \cD$ and furnishes the necessary bijection.
Part (ii) is immediate from the construction.
\end{proof}

The isomorphisms from Lemma~\ref{lem: equivalence isos} are natural in $\cF$ and $\cV$ respectively, and thus we have proven the following.

\begin{thm}\label{thm: equiv of categories}
Let $\mathsf{MW}$ be the category of (not-necessarily-unitary) M\"obius-covariant Wightman CFTs and let $\mathsf{MVA}^+$ be the category of M\"obius vertex algebras equipped with a generating family of quasiprimary vectors.
Let $F:\mathsf{MW} \to \mathsf{MVA}^+$ be the functor constructed on objects in Theorem~\ref{thm: Wightman to VA} and on morphisms in Lemma~\ref{lem: functoriality Wightman to VA}.
Let $G:\mathsf{MVA}^+ \to \mathsf{MW}$ be the functor constructed on objects in Theorem~\ref{thm: VA to Wightman} and on morphisms in Lemma~\ref{lem: functoriality VA to Wightman}.
Then $F$ and $G$, along with the isomorphisms of Lemma~\ref{lem: equivalence isos}, are an equivalence of categories between $\mathsf{MW}$ and $\mathsf{MVA}^+$.
\end{thm}

We close the section with a brief discussion of our results in the context of theories with full conformal symmetry rather than M\"obius symmetry.
The following modification is reasonably straightforward.
We could consider Wightman CFTs on $S^1$ with an energy-momentum tensor, i.e. the family $\cF$ of quasiprimary fields contains a preferred field $T(z)$ such that $L_n := T(e_n)$  satisfy the Virasoro algebra relation with some central charge $c \in \mathbb{C}$ and $L_{-1}, L_0, L_1$ agree with the generators of $\Mob$\footnote{Note also that if we assume the CFT-type condition $ker L_0 = \mathbb{C}\Omega$ then it is enough to assume that $T(e_n)$, $n=-1,0,1$, are the generators of the $\Mob$ symmetry. Then, the Virasoro algebra relations follow by \cite[Thm. 4.10]{Kac98}.}.
Our results give rise to a correspondence between ($\bbN$-graded) conformal vertex algebras and Wightman CFTs with energy-momentum tensor. 

On the other hand, it would also be desirable to have a correspondence between (i) conformal vertex algebras equipped with a generating family of primary vectors, and (ii) diffeomorphism covariant Wightman theories on $S^1$.
Since the action of diffeomorphisms of $S^1$ do not generally fix the vacuum vector, the domain of the Wightman theory would not be spanned by expressions $\varphi_1(f_1) \cdots \varphi_k(f_k)\Omega$, but instead by ones of the form $\varphi_1(f_1) \cdots \varphi_k(f_k)U(\gamma)\Omega$ with $\gamma \in \Diff(S^1)$.
This raises several points to address, including the question of exponentiating representations of the Virasoro algebra to an action of $\Diff(S^1)$ (on some completion of the original space) in a non-unitary setting.
These are non-trivial technical challenges that we do not consider in this article. 
However, we do not see any obstructions to completing such a generalization in the unitary setting, where the theory of exponentiating positive energy representations of the Virasoro algebra is well-developed.

\section{Invariant forms and unitary theories}\label{sec: invariant forms}

In this section we show that the correspondence between Wightman CFTs on $S^1$ and M\"obius vertex algebras constructed in Section~\ref{sec: MVA and WCFT}  is compatible with invariant bilinear forms.
The definition of an invariant bilinear form for a M\"obius vertex algebra is standard (see \cite[\S5.2]{FHL93} and \cite{Li94}).

\begin{defn}
An invariant bilinear form $( \cdot , \cdot)$ on a M\"obius vertex algebra $\cV$ is a bilinear form such that 
\begin{equation}\label{eqn: invariant bilinear}
(Y(v,z)u_1,u_2) = (u_1, Y(\rme^{zL_1}(-z^{-2})^{L_0}v,z^{-1})u_2)
\end{equation}
and 
\begin{equation}\label{eqn: invariant mob}
(L_nu_1, u_2) = (u_1, L_{-n}u_2)
\end{equation}
for all $v,u_1,u_2 \in \cV$.
\end{defn}
It can be convenient to introduce notation for the opposite vertex operator
\begin{equation}\label{eqn: opposite vertex operator}
Y^\opp(v,z) = Y(\rme^{zL_1}(-z^{-2})^{L_0}v,z^{-1}),
\end{equation}
and in this notation the invariance condition becomes 
\[
(Y(v,z)u_1,u_2) = (u_1, Y^\opp(v,z)u_2).
\]

The map $L_n \mapsto -L_{-n}$ extends linearly to a Lie algebra automorphism of $\Lie(\Mob)_{\bbC}$ which leaves $\Lie(\Mob)$ invariant.
Let $d\alpha:\Lie(\Mob) \to \Lie(\Mob)$ be this restriction.
In this notation, the compatibility condition \eqref{eqn: invariant mob} between the invariant bilinear form and the representation $\pi$ of $\Lie(\Mob)$ on $\cV$ becomes
\[
(\pi(f \tfrac{d}{d\vartheta})u_1,u_2) = -(u_1,\pi(d\alpha(f\tfrac{d}{d\vartheta}))u_2).
\]
In order to formulate the correct notion of invariant bilinear form for a Wightman CFT, we must integrate $d\alpha$ to an automorphism $\alpha$ of $\Mob$.
It is straightforward to check that $\alpha$ is given by
\[
(\alpha \gamma)(z) = 1/\gamma(\tfrac{1}{z}).
\]
Indeed, at the level of matrices $\alpha$ is given on $\begin{pmatrix} a & b\\ \overline{b} & \overline{a}\end{pmatrix} \in \mathrm{SU}(1,1)$ (with $\abs{a}^2 - \abs{b}^2 = 1$) by complex conjugation
\[
\alpha \begin{pmatrix} a & b\\ \overline{b} & \overline{a}\end{pmatrix} = \begin{pmatrix} \overline{a} & \overline{b}\\ b & a\end{pmatrix}
\]
and $d\alpha$ is given on $\begin{pmatrix} \rmi c & d\\\overline{d} & -\rmi c \end{pmatrix} \in \mathfrak{su}(1,1)$ (with $c \in \bbR$) by complex conjugation as well
\[
d\alpha \begin{pmatrix} \rmi c & d\\\overline{d} & -\rmi c \end{pmatrix} = \begin{pmatrix} -\rmi c & \overline{d}\\ d & \rmi c \end{pmatrix}.
\]
In particular we have
\begin{equation}\label{eqn: exp of lie algebra alpha}
\exp(d\alpha(f\tfrac{d}{d\vartheta})) = \alpha(\exp(f \tfrac{d}{d\vartheta})).
\end{equation}

We thus have the following notion of invariant bilinear form for a Wightman CFT.

\begin{defn}
Let $(\cF,\cD,U,\Omega)$ be a M\"obius-covariant Wightman CFT on $S^1$.
A jointly $\cF$-strong continuous bilinear form $(\, \cdot \, , \, \cdot \,)$ on $\cD$ is called an invariant bilinear form if
\begin{equation}\label{eqn: invariant Wightman form}
( \varphi(f)\Phi, \Psi) = (\Phi, (-1)^{d_\varphi} \varphi(f\circ \tfrac{1}{z})\Psi)
\end{equation}
for all $\varphi \in \cF$ (with conformal dimension $d_\varphi$), all $f \in C^\infty(S^1)$, and all $\Phi,\Psi \in \cD$, and moreover
\begin{equation}\label{eqn: invariant Wightman mob}
(U(\gamma)\Phi,U(\alpha(\gamma))\Psi) = (\Phi,\Psi)
\end{equation}
for all $\gamma \in \Mob$ and $\Phi,\Psi \in \cD$.
\end{defn}
As in the context of vertex algebras, we can introduce the notion of opposite field
\[
\varphi^\opp(f) := (-1)^{d_\varphi} \varphi(f \circ \tfrac{1}{z})
\]
and the invariance condition \eqref{eqn: invariant Wightman form} then becomes
\[
(\varphi(f)\Phi,\Psi) = (\Phi, \varphi^\opp(f)\Psi).
\]

\begin{thm}[Correspondence between invariant bilinear forms]\label{thm: equivalence of bilinear forms}
Let $(\cF,\cD,U,\Omega)$ be a M\"obius-covariant Wightman CFT on $S^1$ and let $\cV \subset \cD$ be the corresponding M\"obius vertex algebra.
Then 
\begin{enumerate}[i)]
\item Every invariant bilinear form for the Wightman CFT $\cD$ restricts to an invariant bilinear form for the vertex algebra $\cV$.
\item Every invariant bilinear form for the vertex algebra $\cV$ extends uniquely to an invariant bilinear form for the Wightman CFT on $\cD$.
\end{enumerate}
If an invariant form on $\cV$ is nondegenerate, then so is the extension to $\cD$.
Conversely, if an invariant form on $\cD$ is nondegenerate then so is the restriction to $\cV$.
\end{thm}
\begin{proof}
First suppose that $\cD$ is equipped with an invariant bilinear form $( \, \cdot \, , \, \cdot \,)$.
Let $X \in \Lie(\Mob)$, let $\gamma_t=\exp(tX) \in \Mob$, and let 
\[
\rho_t = \alpha(\exp(t X))=\exp(t d\alpha(X)).
\]
For $u_1,u_2 \in \cV$ we have
\[
(U(\gamma_t)u_1,u_2) = (u_1,U(\rho_{-t})u_2).
\]
Differentiating and evaluating at $t=0$ (as in the proof of Lemma~\ref{lem: Lie Mob covariance}) yields
\[
(\pi(X)u_1,u_2) = -(u_1,\pi(d\alpha(X))u_2),
\]
as required.
Now let $S \subset \cV$ be the set of quasiprimary generators corresponding to $\cF$.
For $v \in S$ we have
\[
(Y(v,f)u_1,u_2) = (-1)^{d_v} (u_1, Y(v,f \circ \tfrac{1}{z})u_2)
\]
and in particular at the level of modes $v_n \in \mathrm{End}(\cV)$ we have
\[
(v_n u_1, u_2) = (-1)^{d_v} (u_1,v_{-n} u_2).
\]
Hence for $v \in S$ we have
\[
(Y(v,z)u_1,u_2) = (u_1, Y^\opp(v,z)u_2).
\]
This extends to all $v \in \cV$ by Lemma~\ref{lem: invariant form for generators} below, and we have established (i).

Now conversely suppose that $\cV$ is equipped with an invariant bilinear form which we denote $(\, \cdot \, , \, \cdot \,)_{\cV}$.
Note that a $\cF$-strongly continuous extension of such a form on $\cV$ to a bilinear form on $\cD$ is unique, and so we must only show existence.
Recall from Proposition~\ref{prop: canonical embedding of D} that $\cD$ comes naturally embedded in $\widehat \cV = \prod_{n=0}^\infty \cV(n)$.
The bilinear form on $\cV$ naturally extends to a pairing of $\cV$ and $\widehat \cV$.
First, we claim that for $\varphi_1, \ldots, \varphi_k$, $\psi_1, \ldots, \psi_\ell \in \cF$ and $f_1,\ldots, f_k, g_1, \ldots, g_\ell \in C^\infty(S^1)$ we have
\begin{equation}\label{eqn: left and right definitions of D form}
(\psi^\opp(g_\ell) \cdots \psi^\opp(g_1)\varphi_1(f_1) \cdots \varphi_k(f_k)\Omega,\Omega)_{\widehat \cV, \cV} = 
(\Omega,\varphi^\opp(f_k) \cdots \varphi^\opp(f_1)\psi(g_1) \cdots \psi(g_\ell)\Omega)_{\cV,\widehat \cV}.
\end{equation}
Indeed these agree when $f_i,g_j \in \bbC[z^{\pm 1}]$ since the form is invariant for $\cV$, and this identity extends to all smooth functions by continuity.

With this in mind, we wish to define a bilinear form on $\cD$ by extending linearly the prescription
\begin{equation}\label{eqn: invariant form on D def}
(\varphi_1(f_1) \cdots \varphi_k(f_k)\Omega, \psi_1(g_1) \cdots \psi_\ell(g_\ell)\Omega)_\cD :=
(\psi^\opp(g_\ell) \cdots \psi^\opp(g_1)\varphi_1(f_1) \cdots \varphi_k(f_k)\Omega,\Omega)_{\widehat \cV, \cV},
\end{equation}
but we must first check that this is well-defined.
Suppose that for some collection of Wightman fields $\varphi_{i,j} \in \cF$ and smearing functions $f_{i,j} \in C^\infty(S^1)$ we have
\[
0 = \sum_i \varphi_{1,j}(f_{1,j}) \cdots \varphi_{k_j,j}(f_{k_j,j})\Omega.
\]
Then for all $\psi_1, \ldots, \psi_\ell \in \cF$ and $g_1, \ldots, g_\ell \in C^\infty(S^1)$
\[
0 = \sum_j \psi^\opp(g_\ell) \cdots \psi^\opp(g_1)\varphi_{1,j}(f_{1,j}) \cdots \varphi_{k_j,j}(f_{k_j,j})\Omega,
\]
and thus
\[
0 = \sum_j (\psi^\opp(g_\ell) \cdots \psi^\opp(g_1)\varphi_{1,j}(f_{1,j}) \cdots \varphi_{k_j,j}(f_{k_j,j})\Omega,\Omega)_{\widehat \cV,\cV}.
\]
This shows that the prescription \eqref{eqn: invariant form on D def} is well-defined in the first input.
We may repeat the above argument (invoking \eqref{eqn: left and right definitions of D form}) to show that it is also well-defined in the second input, and we conclude that \eqref{eqn: invariant form on D def} extends to a well-defined bilinear form.
As \eqref{eqn: invariant form on D def} is continuous in the functions $f_j$ and $g_j$, the bilinear form on $\cD$ is jointly $\cF$-strong continuous, as required.

Finally, we need to check that $( \, \cdot \, , \cdot \,)_{\cD,\cD}$ is compatible with the representation $U$ of $\Mob$.
Let $X \in \Lie(\Mob)$, let $\gamma_t = \exp(tX)$ and recall that $\alpha(\gamma_t) = \exp(t d\alpha(X))$.
From the proof of Lemma~\ref{lem: LieMob rep extends to D} we have for $\Phi \in \cD$
\[
\frac{d}{dt} U(\gamma_t)\Phi = \left. \frac{d}{ds}U(\gamma_{t+s})\Phi \right|_{s=0}
= \left. \frac{d}{ds}U(\gamma_s)U(\gamma_{t})\Phi \right|_{s=0} = \pi(X) U(\gamma_t)\Phi,
\]
with the derivative taken in the $\cF$-strong topology on $\cD$.
Similarly
\[
\frac{d}{dt} U(\alpha(\gamma_{t}))\Psi = \pi(d\alpha(X))U(\alpha(\gamma_{t}))\Psi.
\]
Hence by the joint continuity of the bilinear form we have
\begin{align*}
&\frac{d}{dt} (U(\gamma_t)\Phi,U(\alpha(\gamma_{t}))\Psi)_\cD \\
&= (\pi(X)U(\gamma_t)\Phi,U(\alpha(\gamma_{-t}))\Psi)_\cD + (U(\gamma_t)\Phi,\pi(d\alpha(X))U(\alpha(\gamma_{t}))\Psi)_\cD\\ 
&= 0.
\end{align*}
In the last equality we used the fact that $(\pi(X)u_1,u_2) = -(u_1,\pi(d\alpha(X))u_2)$ for $u_j \in \cV$, which extends to vectors in $\cD$ by the $\cF$-strong continuity of $\pi(X)$ and $\pi(d\alpha(X))$.
Hence the above expression is independent of $t$, and as the exponential map $\Lie(\Mob) \to \Mob$ is surjective we then have
\[
(\Phi,\Psi)_{\cD} = (U(\gamma)\Phi,U(\alpha(\gamma))\Psi)_\cD
\]
for all $\gamma \in \Mob$ and $\Phi,\Psi \in \cD$, as required.

Finally, we address nondegeneracy.
Recall that $\cD$ embeds naturally in $\widehat \cV$ and that the bilinear form on $\cD$ is compatible with the pairing of $\cV$ and $\widehat \cV$.
If the bilinear form on $\cV$ is nondegenerate and $\Phi=(\Phi_n)_{n \ge 0} \in \cD$ (with $\Phi_n \in \cV(n)$) is non-zero, then $\Phi_n \ne 0$ for some $n$, and thus there exists $v \in \cV(n)$ such that
\[
(\Phi,v) = (\Phi_n,v) \ne 0.
\]
A similar argument shows that the right-kernel of the form is zero, and so the form on $\cD$ is nondegenerate.

Conversely, assume that the form on $\cD$ is nongenerate. By M\"obius invariance, its restriction to $\cV$ is the direct sum $(v,u) = \sum_n (v,u)$. 
Let $v \in \cV$ with $v \neq 0$.
Then there exists $\Phi =(\Phi_n)_{n \ge 0} \in \cD$ with
\[
0 \ne (v,\Phi) = \sum_n (v_n,\Phi_n).
\]
Hence there must be some $n$ such that
$(v_n,\Phi_n) \neq 0$, and so the left-kernel of the form on $\cV$ is zero.
A similar argument shows that right-kernel is zero as well.
\end{proof}

We used the following fact in the proof of Theorem~\ref{thm: equivalence of bilinear forms}.
\begin{lem}\label{lem: invariant form for generators}
Let $\cV$ be a M\"obius vertex algebra equipped with a bilinear form $( \, \cdot \, , \, \cdot \, )$, and let $S \subset \cV$ be a set of vectors that generate $\cV$.
Suppose that the invariance condition
\[
(Y(v,z)u_1,u_2) = (u_1, Y^\opp(v,z)u_2)
\]
holds for $v \in S$ and $u_1,u_2 \in \cV$, and also that
\[
(L_nu_1, u_2) = (u_1, L_{-n}u_2)
\]
for all $u_1,u_2 \in \cV$.
Then the invariance condition holds for all $v \in \cV$ (i.e.\! the form is an invariant bilinear form for $\cV$).
\end{lem}
\begin{proof}
There is a (generalized) $\cV$-module structure on the restricted dual $\cV'$ whose state-field correspondence $Y'(v,z)$ is characterized by
\[
(Y^\opp(v,z)u, u')_{\cV,\cV'} = (u,Y'(v,z)u')_{\cV,\cV'}
\]
for all $u,v \in \cV$ and $u' \in \cV'$.
This contragredient module structure was first studied in \cite[\S5.2]{FHL93} and described further in our context with infinite-dimensional weight spaces in the paragraphs following \cite[Lem.\! 2.22]{HLZI}.
If $f:\cV \to \cV'$ is the map $f(u) = (u, \, \cdot \,)$, then our hypothesis implies that $f Y(v,z) = Y'(v,z)f$ for all $v \in S$, or at the level of modes $fv_{(n)} = v'_{(n)}f$ for all $v \in S$ and $n \in \bbZ$.
This intertwining condition extends to all $v \in \cV$ by the Borcherds product formula (for $\cV$ and for $\cV'$), and we conclude that the bilinear form is invariant.
\end{proof}

It was shown in  \cite[Prop. 5.3.6]{FHL93} that every nondegenerate invariant bilinear form on a vertex operator algebra is symmetric.
Later it was observed in \cite[Prop. 2.6]{Li94} that the proof does not use the hypothesis of nondegeneracy, and further examination of the proof in \cite{FHL93} shows that the proof also goes through for M\"obius vertex algebras as defined in this article (that is, allowing for infinite-dimensional $L_{0}$-weight spaces and only using M\"obius symmetry rather than Virasoro).
In light of Theorem~\ref{thm: equivalence of bilinear forms}, we have the same result for Wightman CFTs.

\begin{cor}
Every invariant bilinear form on a Wightman CFT is symmetric.
\end{cor}

We now turn our attention to theories equipped with invariant sesquilinear forms (which we call involutive structures), noting this includes unitary theories.
In order to do this we will need to introduce antilinear homomorphisms of M\"obius vertex algebras and M\"obius-covariant Wightman CFTs.
Let $\cV$ and $\tilde \cV$ be M\"obius vertex algebras, with vacuum vectors $\Omega$ and $\tilde \Omega$ and representations $L_n$ and $\tilde L_n$ of $\Lie(\Mob)_{\bbC}$, respectively.
Then an antilinear map $g:\cV \to \tilde \cV$ is called a homomorphism if $g(\Omega) = \tilde \Omega$ and
\[
g v_{(m)} = (gv)_{(m)}g, \qquad \mbox{ and } \qquad  gL_n = \tilde L_n g
\]
for all $v \in \cV$, $m \in \bbZ$ and $n=-1,0,1$.

On the Wightman side, if $(\cF, \cD, U, \Omega)$ and $(\tilde \cF, \tilde \cD, \tilde U, \tilde \Omega)$ are M\"obius-covariant Wightman CFTs, an antilinear homomorphism $\cF \to \tilde \cF$ is an antilinear map $g:\cD \to \tilde \cD$ and a function $g_*:\cF \to \tilde \cF$ such that $g(\Omega) = \tilde \Omega$ and
\[
g\varphi(f) = (g_*\varphi)(\overline{f} \circ \tfrac{1}{z})g, \qquad \mbox{ and } \qquad gU(\gamma) = \tilde U(\alpha(\gamma))g
\]
for all $\varphi \in \cF$, $f \in C^\infty(S^1)$ and $\gamma \in \Mob$ (and we recall $\alpha(\gamma)(z) = 1/\gamma(\tfrac{1}{z})$)\footnote{It may be surprising that the condition $gL_n = \tilde L_n g$ for vertex algebras corresponds to $gU(\gamma) = \tilde U(\alpha(\gamma))g$ for Wightman CFTs. Note that due to the antilinearity of $g$, the relation $gL_n = \tilde L_n g$ does not imply that $g$ intertwines the representations of $\Lie(\Mob)$. In fact we have $g\pi(X) = \tilde \pi(d\alpha(X)) g$ for $X \in \Lie(\Mob)$.}.
We have
\[
\tilde U(\gamma) (g_*\varphi)(f) = g U(\alpha(\gamma))\varphi(\overline{f}\circ\tfrac{1}{z}),
\]
where $\overline{f}$ denotes the pointwise complex conjugate.
Just as we demonstrated in Lemmas~\ref{lem: functoriality Wightman to VA} and \ref{lem: functoriality VA to Wightman}, one can show that antilinear homomorphisms of M\"obius vertex algebras extend uniquely to antilinear homomorphisms of Wightman CFTs, and conversely antilinear homomorphisms of Wightman CFTs restrict to antilinear homomorphisms of M\"obius vertex algebras. 

\begin{lem}\label{lem: functoriality antilinear homomorphisms}
Let $(\cF,\cD,U,\Omega)$ and $(\tilde \cF, \tilde \cD, \tilde U, \tilde \Omega)$ be two M\"obius-covariant Wightman CFTs and let $\cV \subset \cD$ and $\tilde \cV \subset \tilde \cD$ be the corresponding M\"obius vertex algebras with respective generating sets $S$ and $\tilde S$.
\begin{enumerate}[i)]
\item If $(g_*,g):\cF \to \tilde \cF$ is an antilinear homomorphism then $g(\cV) \subset \tilde \cV$ and $g|_{\cV}$ is an antilinear homomorphism of M\"obius vertex algebras satisfying $g(S) \subset \tilde S$.
\item If $g:\cV \to \tilde \cV$ is an antilinear homomorphism of M\"obius vertex algebras such that $g(S) \subset \tilde S$, then there is a unique antilinear homomorphism $(h_*,h):\cF \to \tilde \cF$ such that $h|_\cV = g$.
\end{enumerate}
\end{lem}

We omit the proof of Lemma~\ref{lem: functoriality antilinear homomorphisms} which is essentially identical to Lemma~\ref{lem: functoriality Wightman to VA} and \ref{lem: functoriality VA to Wightman} (once we observe as above that an antilinear vertex algebra homomorphism satisfies $g\pi(X) = \tilde \pi(d\alpha(X))g$ for $X \in \Lie(\Mob)$).

Recall that an antilinear map $g$ is said to preserve a sesquilinear form if $\ip{g\Phi, g\Psi} = \overline{\ip{\Phi,\Psi}}$ for all vectors $\Phi,\Psi$, and that a sesquilinear form is said to be (Hermitian) symmetric if $\ip{\Phi,\Psi}=\overline{\ip{\Psi,\Phi}}$.

\begin{defn}\label{defn: involutive vertex algebra}
An \textbf{involutive M\"obius vertex algebra} is a M\"obius vertex algebra $\cV$ equipped with a sesquilinear form $\ip{ \, \cdot \, , \, \cdot \,}$ and an antilinear automorphism $\theta:\cV \to \cV$ which is involutive $(\theta^2=\mathrm{id}_\cV)$ and preserves the sesquilinear form, and such that $\ip{ \, \cdot \, , \theta\, \cdot \,}$ is an invariant bilinear form.
An involutive M\"obius vertex algebra is called \textbf{unitary} if the sesquilinear form is an inner product that is normalized so that $\ip{\Omega,\Omega}=1$.
\end{defn}

We use the convention that sesquilinear forms are linear in the first variable, and require that homomorphisms of M\"obius vertex algebras commute with the operators $L_n$.
The condition that $\ip{ \, \cdot \, , \theta\, \cdot \,}$ is an invariant bilinear form is equivalent to having
\begin{equation}\label{eqn: rephrased invariant sesquilinear form}
\ip{Y(v,z)u_1,u_2} = \ip{u_1,Y^\opp(\theta v, \bar z)u_2} \quad \mbox{ and } \quad \ip{L_nu_1, u_2} = \ip{u_1, L_{-n}u_2}
\end{equation}
for all $u_1,u_2,v \in \cV$ and $n=-1,0,1$, where $z$ is a formal complex variable,
i.e.\! $\ip{ \, \cdot \, , \bar z\, \cdot \,} = z\ip{ \, \cdot \, , \cdot \,}$.

We sometimes refer to the sesquilinear form from Definition~\ref{defn: involutive vertex algebra} as an invariant sesquilinear form, omitting reference to the involution $\theta$.

\begin{rem}\label{rmk: invariant sesquilinear redundancy}
The sesquilinear forms from Definition~\ref{defn: involutive vertex algebra} are automatically Hermitian symmetric as a consequence of the fact that invariant bilinear forms are symmetric.
If the sesquilinear form is nondegenerate then the requirement that $\theta$ be involutive is redundant and the automorphism $\theta$ is uniquely determined by the sesquilinear form (the proof is exactly as in \cite[Prop. 5.1]{CKLW18} for inner products).
\end{rem}

We now turn our attention to invariant sesquilinear forms on Wightman CFTs.

\begin{defn}\label{defn: Wightman involutive structure}
An \textbf{involutive M\"obius-covariant Wightman CFT on $S^1$} is a Wightman CFT $(\cF,\cD,U,\Omega)$ along with a jointly $\cF$-strong continuous sesquilinear form $\ip{\, \cdot \, , \, \cdot \,}$ on $\cD$ and an involutive automorphism $(\theta_*,\theta)$ of $\cF$ such that $\langle \theta u, \theta v \rangle = \langle u,v \rangle$ and such that $\ip{\,\cdot\, , \, \theta \cdot \,}$ is an invariant bilinear form. 
An involutive M\"obius-covariant Wightman CFT is called \textbf{unitary} if the sesquilinear form is an inner product which is normalized so that $\ip{\Omega,\Omega}=1$.
\end{defn}

One can think of $\theta$ as generalising the notion of a PCT operator associated with a positive-definite form. As with vertex algebras, we sometimes refer to the sesquilinear form of Definition~\ref{defn: Wightman involutive structure} as an invariant sesquilinear form, omitting reference to the involution.

If we write $\varphi^\dagger = (-1)^{d_\varphi}\theta_*\varphi$, then the condition that $\ip{\,\cdot\, , \, \theta \cdot \,}$ is an invariant bilinear form is equivalent to
\begin{equation}\label{eqn: compatibility sesquilinear and wightman field}
\ip{\varphi(f)\Phi,\Psi} = \ip{\Phi,\varphi^\dagger(\overline{f})\Psi} \quad \mbox{ and } \quad \ip{U(\gamma)\Phi, U(\gamma)\Psi} = \ip{\Phi,\Psi}
\end{equation}
for all $\Phi,\Psi \in \cD$, $\varphi \in \cF$, and $\gamma \in \Mob$\footnote{Note that this is a slight departure from \cite{RaymondTanimotoTener22}, where we required that $\cF$ be invariant under the involution $\dagger$ rather than $\theta_*$. In the present setting we find the updated definition to be more natural, as the involution $\dagger$ of fields typically does not correspond to an antilinear automorphism of the Wightman CFT.}.
Here, as before, $\overline{f}$ denotes the pointwise complex conjugate of the function $f$.

As with involutive vertex algebras (Remark~\ref{rmk: invariant sesquilinear redundancy}), the sesquilinear form of an involutive Wightman CFT is automatically Hermitian symmetric.

\begin{thm}[Equivalence of involutive and unitary structures] \label{thm: equivalence of involutive structures}
Let $(\cF, \cD, U, \Omega)$ be a M\"obius-covariant Wightman CFT, and let $\cV \subset \cD$ be the corresponding M\"obius vertex algebra equipped with a set $S$ of quasiprimary generators.
Then we have the following.
\begin{enumerate}[i)]
\item If $\cD$ is equipped with a sesquilinear form and involution $(\theta_*,\theta)$ making it into an involutive Wightman CFT, then the sesquilinear form and involution $\theta$ restrict to an involutive structure on the vertex algebra $\cV$.
The set $S \subset \cV$ of quasiprimary generators is invariant under $\theta$.
\item If $\cV$ is equipped with a sesquilinear form and involution $\theta$ making it into an involutive vertex algebra and $S$ is invariant under $\theta$, then there is a unique involution $\theta_*$ of $\cF$ and unique extensions of the sesquilinear form and $\theta$ to $\cD$ making $\cF$ into an involutive Wightman CFT.
\end{enumerate}
If the sesquilinear form is nondegenerate on $\cD$ then it remains nondegenerate on $\cV$, and similarly if the form is nondegenerate on $\cV$ so is the extension to $\cD$.
Moreover unitary structures on $\cD$ correspond to unitary structures on $\cV$, and vice versa.
\end{thm}
\begin{proof}
First consider an involutive structure $(\theta_*,\theta)$ on $\cF$. 
Then $\theta$ restricts to an antilinear involutive M\"obius vertex algebra automorphism $\theta|_\cV : \cV \to \cV$ by Lemma~\ref{lem: functoriality antilinear homomorphisms}.
By Theorem~\ref{thm: equivalence of bilinear forms}, the invariant bilinear form $\ip{\, \cdot \, , \, \theta \cdot \,}$ restricts to an invariant bilinear form on $\cV$, and it follows that $\ip{\, \cdot \, , \, \cdot \,}$ and $\theta|_{\cV}$ yield an involutive structure on $\cV$.
If $\varphi \in \cF$ corresponds to the state $v \in \cV(d)$, then
\[
\theta v = \theta \varphi(e_{-d})\Omega = (\theta_* \varphi)(e_{-d})\Omega. 
\]
Hence the Wightman field $\theta_*\varphi \in \cF$ corresponds to $\theta v$ and $S$ is $\theta$-invariant (we have used here the observation that $\theta_*$ preserves the conformal dimension of fields).

For the other direction, suppose that we have an involutive structure on $\cV$ corresponding to an involution $\theta$ and sesquilinear form $\ip{\, \cdot \, , \, \cdot \, }$.
Then the invariant bilinear form $\ip{\, \cdot \, , \theta \, \cdot \,}$ extends uniquely to an invariant bilinear form $( \, \cdot \, , \, \cdot \,)$ on $\cD$.
By Lemma~\ref{lem: functoriality antilinear homomorphisms} we may uniquely extend $\theta$ to an antilinear automorphism $(\theta_*,\theta)$ of $\cF$.
This extension is $\cF$-strong continuous, and thus the sesquilinear form $(\, \cdot \, , \,\theta \cdot \,)$ is $\cF$-strong continuous as well.
This sesquilinear form, along with $(\theta_*, \theta)$, yield an involutive structure on $\cF$ as required.

The proof of equivalence of nondegeneracy is straightforward (as in the proof of Theorem~\ref{thm: equivalence of bilinear forms}), and the equivalence of unitarity is immediate.
\end{proof}

For unitary Wightman CFTs domain $\cD$ can be equipped with the norm topology coming from the inner product. 
This leads to a number of analytic questions, which are discussed further in \cite{RaymondTanimotoTener22}.

\begin{rem}
Suppose that $(\cF,\cD,U,\Omega)$ is a Wightman CFT such that $\cD$ is equipped with an inner product and $\cF$ is equipped with an involution $\dagger$.
The definition of a (unitary) Wightman CFT given in \cite{RaymondTanimotoTener22} required only that the compatibility conditions \eqref{eqn: compatibility sesquilinear and wightman field} hold, with no mention of the PCT operator $\theta$.
However, under these assumptions one may show that there exists a unique $\theta$ making the associated vertex algebra $\cV$ into a unitary vertex algebra, arguing as in \cite[Thm. 3.11]{RaymondTanimotoTener22} based on \cite[Thm. 5.16]{CKLW18}\footnote{The hypothesis that $\dim \cV(0) = 1$ required in \cite[Thm. 5.16]{CKLW18} is not needed, as shown in \cite[\S3.4]{CarpiGaudioHillier23ax}, and the condition $\dim \cV(n) < \infty$ is also not needed.}.
One may then extend $\theta$ to an antilinear involution $(\theta_*,\theta)$ of $\cF$ by Lemma~\ref{lem: functoriality antilinear homomorphisms}, making $(\cF,\cD,U,\Omega)$ into a unitary Wightman CFT as defined in this article.
We note that for a general sesquilinear form on $\cD$, an involution $\dagger$ satisfying \eqref{eqn: compatibility sesquilinear and wightman field} does not necessarily correspond to an involutive structure.
Indeed, in the extreme example where the sesquilinear form is identically zero, the compatibility conditions \eqref{eqn: compatibility sesquilinear and wightman field} impose no constraint on the involution $\dagger$, but not every set-theoretic involution of $\cF$ corresponds to an involutive structure.
It is possible that the conditions \eqref{eqn: compatibility sesquilinear and wightman field} are sufficient to reconstruct the PCT operator $\theta$ when the sesquilinear form is nondegenerate, but we do not address that question here.
\end{rem}

\appendix

\section[The Reeh-Schlieder theorem for non-unitary Wightman conformal field theories]{The Reeh-Schlieder theorem for non-unitary\\ Wightman conformal field theories} \label{sec: reeh schlieder}

In this section we work with rotation-covariant Wightman QFTs $(\cF,\cD,U,\Omega)$ on $S^1$, which differ from M\"obius-covariant Wightman CFTs (Definition~\ref{def: Wightman}) only in that the symmetry $U$ is only a representation of the rotation subgroup $\Rot(S^1) \subset \Mob$, and accordingly the covariance condition (W\ref{itm: W Mob covariance}) is weakened to only require covariance for rotations.
We write either $R_z$ or $R_\vartheta$ for rotation by $z=\rme^{\rmi\vartheta}$.

For $I \subset S^1$ an interval (i.e.\! $I$ is a connected open non-empty proper subset), we let $\cP(I) \subset \cL(\cD)$ be the algebra generated by smeared fields $\varphi(f)$ with $\supp f \subset I$.
The goal of this section is to establish the Reeh-Schlieder theorem for the theory $\cF$, which says that the vacuum vector $\Omega$ is cyclic and separating for the algebras $\cP(I)$.
Recall that $\Omega$ is \textbf{cyclic} for an algebra $\cP \subset \cL(\cD)$ (with respect to a certain topology on $\cD$) if $\cP \Omega$ is dense in $\cD$, and \textbf{separating} for $\cP$ if the only $X \in \cP$ such that $X\Omega = 0$ is $X=0$.
The analogous statement for unitary Wightman quantum field theories on higher-dimensional spacetimes is well-known (see \cite[\S4.2]{StreaterWightman64} and \cite{ReehSchlieder}).
We give here a proof of the Reeh-Schlieder theorem in our current (not necessarily unitary) context.

Let $D$ be the open unit disk in $\bbC$, and $\overline{D}$ its closure.
We denote by $A(\overline{D})$ the space of continuous $\bbC$-valued functions on $\overline{D}$ that are holomorphic on the interior $D$.
By the maximum principle $A(\overline{D})$ embeds as a closed subspace of $C(S^1)$, and we give $A(\overline{D})$ the norm inherited from $C(S^1)$.

\begin{lem}\label{lem: reeh schlieder lemma}
Let $\cF$ be a rotation-covariant Wightman CFT.
Fix $\varphi_1, \ldots, \varphi_k \in \cF$, and let $\lambda \in \cD_\cF^*$.
Let $f_1, \ldots, f_k \in C^\infty(S^1)$ and let $z_1, \ldots, z_k \in S^1$.
Then for each $j = 1,\cdots, k$, the map
\begin{equation}\label{eqn: reeh schlieder function}
z_j \mapsto \lambda\big(U(R_{z_1}) \varphi_1(f_1) U(R_{z_2}) \varphi_2(f_2) \cdots U(R_{z_k}) \varphi_k(f_k)\Omega\big)
\end{equation}
lies in $A(\overline{D})$.
\end{lem}
\begin{proof}
When the functions $f_j$ are all Laurent polynomials, the expression \eqref{eqn: reeh schlieder function} is a polynomial in the $z_i$ and the conclusion follows.
We now consider the general case.

By rotation covariance we have
\[
\lambda\big(U(R_{z_1}) \varphi_1(f_1) U(R_{z_2}) \cdots U(R_{z_k}) \varphi_k(f_k)\Omega\big)
=
\lambda\big(\varphi_1(\beta_{d_1}(R_{w_1}) f_1)  \cdots  \varphi_k(\beta_{d_k}(R_{w_k})f_k)\Omega\big)
\]
where $w_j = z_1z_2 \cdots z_j$, and $d_j$ is the conformal dimension of $\varphi_j$.
Given arbitrary smooth $f_j$, choose sequences of Laurent polynomials $f_{j,n}$ such that $f_{j,n} \to f_j$ in $C^\infty(S^1)$.
As in Section~\ref{sec: VOA to Wightman}, let $H^N(S^1)$ be the Sobolev space corresponding to a number $N>0$, and recall that the topology on $C^\infty(S^1)$ is generated by the Sobolev norms $\norm{\,\cdot\,}_N$.
Since $\beta_d(R_w)$ acts as a unitary on $H^N(S^1)$, we have convergence in each $H^N(S^1)$
\[
\lim_{n \to \infty} \beta_{d_j}(R_{w})f_{j,n} = \beta_{d_j}(R_w)f_j
\]
that is uniform in $w$.

Since $\lambda \in \cD_\cF^*$, expressions
\begin{equation}\label{eqn: rs inner product expression}
\lambda\big(\varphi_1(f_1) \cdots \varphi_k(f_k)\Omega\big)
\end{equation}
are jointly continuous as maps $C^\infty(S^1)^k \to \bbC$.
Hence we may choose a positive number $N$ such that \eqref{eqn: rs inner product expression} is jointly continuous $H^N(S^1)^k \to \bbC$ (i.e.\! it is a bounded multilinear map).
It follows that
\[
\lim_{n \to \infty} \lambda\big(\varphi_1(\beta_{d_1}(R_{w_1}) f_{1,n})  \cdots  \varphi_k(\beta_{d_k}(R_{w_k})f_{k,n})\Omega\big) = \lambda\big(\varphi_1(\beta_{d_1}(R_{w_1}) f_1)  \cdots  \varphi_k(\beta_{d_k}(R_{w_k})f_k)\big)
\]
uniformly in $z_1, \ldots, z_k$.
As each map 
\begin{align*}
z_j &\mapsto \lambda\big(U(R_{z_1}) \varphi_1(f_{1,n}) U(R_{z_2}) \varphi_2(f_2) \cdots U(R_{z_k}) \varphi_k(f_{k,n})\Omega\big) \\
& = \lambda\big(\varphi_1(\beta_{d_1}(R_{w_1}) f_{1,n})  \cdots  \varphi_k(\beta_{d_k}(R_{w_k})f_{k,n})\Omega\big)
\end{align*}
lies in $A(\overline{D})$ and $A(\overline{D})$ is a closed subspace of $C(S^1)$, the map \eqref{eqn: reeh schlieder function} lies in $A(\overline{D})$ as claimed.
\end{proof}

\begin{lem}\label{lem: RS vanish on interval vanish everywhere}
Let $\cF$ be a rotation-covariant Wightman CFT on $S^1$ with domain $\cD$, and let $I \subset S^1$ be an interval.
Let $\lambda \in \cD_\cF^*$, and suppose $\lambda(X\Omega) = 0$ for all $X \in \cP(I)$.
Then $\lambda = 0$.
\end{lem}
\begin{proof}
Fix $\varphi_1, \ldots, \varphi_k \in \cF$, so that
\begin{equation}\label{eqn: P(I) perp}
\lambda\big(\varphi_1(f_1) \cdots \varphi_k(f_n)\Omega\big) = 0
\end{equation}
whenever $\supp(f_j) \subset I$ for $j=1,\ldots,k$.
Fix $f_1,\ldots,f_k$ supported in $I$, and consider the function $F_k:S^1 \to \bbC$ given by
\[
F_k(z) = \lambda\big(\varphi_1(f_1) \cdots \varphi_{k-1}(f_{k-1})U(R_z)\varphi_k(f_k)\Omega\big).
\]
We have $F_k \in A(\overline{D})$ by Lemma~\ref{lem: reeh schlieder lemma}.
Moreover, by rotation covariance $F_k$ vanishes on a small interval of $S^1$ about $1$ (note that $\supp(f)$ is closed and the interval $I$ is open, so that $I$ contains a neighborhood of $\supp(f_k)$).
Thus by the Schwarz reflection principle we have $F_k = 0$ identically, and restricting to $z \in S^1$ we have
\[
0 = F_k(z) = \lambda\big(\varphi_1(f_1) \cdots \varphi_{k-1}(f_{k-1})\varphi_k(\beta_d(R_z)f_k)\Omega\big)
\]
for all $z \in S^1$.
Hence \eqref{eqn: P(I) perp} holds whenever $f_1, \ldots, f_{k-1}$ are supported in $I$, and $f_k$ is supported in any interval of length $\abs{I}$.
Using a partition of unity, it follows that \eqref{eqn: P(I) perp} holds for arbitrary $f_k$.

We now repeat the above argument.
As before, we may show that the function
\[
z \mapsto \lambda\big(\varphi_1(f_1) \cdots U(R_z) \varphi_{k-1}(f_{k-1}) \varphi_k(f_k)\Omega\big)
\]
vanishes identically on $S^1$, and from there deduce that \eqref{eqn: P(I) perp} holds whenever $f_1, \ldots, f_{k-2}$ are supported in $I$, and $f_{k-1},f_k$ are arbitrary.
Repeatedly applying this argument, we see that \eqref{eqn: P(I) perp} holds for all $f_1, \ldots, f_k \in C^\infty(S^1)$, which means $\lambda=0$ by the vacuum axiom of a Wightman CFT.
\end{proof}

\begin{cor}[Reeh-Schlieder theorem]\label{cor: reeh schlieder}
Let $\cF$ be a rotation-covariant Wightman CFT on $S^1$ with domain $\cD$.
For $I \subset S^1$ an interval we let $\cP(I) \subset \cL(\cD)$ be the subalgebra generated by $\varphi(f)$ with $\varphi \in \cF$ and $\supp(f) \subset I$.
Then
\begin{enumerate}[i)]
\item $\Omega$ is cyclic for $\cP(I)$ with respect to the $\cF$-strong topology on $\cD$, i.e.\! $\cP(I)\Omega$ is $\cF$-strongly dense in $\cD$.
\item $\Omega$ is separating for $\cP(I)$, i.e.\! if $X \in \cP(I)$ and $X\Omega = 0$ then $X=0$.
\end{enumerate}
\end{cor}
\begin{proof}
For part (i), recall from Remark~\ref{rem: weak of F strong is F weak} that $\cD_\cF^*$ is precisely the dual space of $\cD$ equipped with the $\cF$-strong topology.
By Lemma~\ref{lem: RS vanish on interval vanish everywhere} the closed subspace $\overline{\cP(I)\Omega}$ is annihilated only by the zero functional, and so by the Hahn-Banach theorem (for locally convex topological vector spaces) we must have $\overline{\cP(I)\Omega}=\cD$.

For part (ii), observe that by the locality axiom of a Wightman theory the operator $X$ vanishes on $\cP(I')\Omega$, where $I'$ is the interval complementary to $I$.
By Lemma~\ref{lem: continuity of field operators} the operator $X:\cD \to \cD$ is $\cF$-strongly continuous, 
and hence by part (i) we have $X=0$.
\end{proof}

\section{Topological vector spaces}\label{app: TVS and LCS}

In this section we supplement the discussion of the topology on the domain $\cD$ of a Wightman field theory by giving additional definitions, details, and references regarding topological vector spaces and locally convex spaces.
We refer readers to the textbooks \cite{NariciBeckenstein11,Treves67} for further reading.
All vector spaces in this section are assumed to be over the field of complex numbers.

A \textbf{topological vector space} is a vector space $V$ equipped with a \textbf{vector topology}, which is a topology such that the addition map $V \times V \to V$ and the scalar multiplication map $\bbC \times V \to V$ are continuous.
Vector topologies are not necessarily Hausdorff by definition, although we will primarily be interested in Hausdorff topological vector spaces.

A \textbf{seminorm} on a vector space $V$ is a map $p:V \to \bbR_{\ge 0}$ such that $p(u+v) \le p(u)+p(v)$ and $p(\alpha u) = \abs{\alpha} p(u)$ for all $u,v \in V$ and $\alpha \in \bbC$.
Given a set of seminorms on $V$, the corresponding \textbf{seminorm topology} is the coarsest topology on $V$ making all of the seminorms continuous.
Seminorm topologies are always vector topologies, but not every vector topology is a seminorm topology.
A \textbf{locally convex space} is a topological vector space whose topology is a seminorm topology corresponding to some set of seminorms.
Equivalently, a locally convex space is a topological vector space such that there exists a neighborhood basis of the origin consisting of convex sets \cite[Thm. 5.5.2]{NariciBeckenstein11}.
Every Hausdorff topological vector space $V$ has a unique \textbf{completion} $\widehat V$ \cite[\S5]{Treves67}, and the completion of a locally convex space is locally convex \cite[Thm. 5.11.5]{NariciBeckenstein11}.
We note that finite-dimensional Hausdorff topological vector spaces are complete \cite[Thm. 4.10.3]{NariciBeckenstein11}, as are products of complete topological vector spaces.
Every continuous linear map $T:U \to V$ of Hausdorff topological vector spaces extends continuously to a map $\widehat T: \widehat U \to \widehat V$ \cite[Thm. 5.2]{Treves67}.

Locally convex spaces play an important role in functional analysis because the Hahn-Banach theorem holds for them.
In particular, the continuous linear functionals on a locally convex Hausdorff space separate points.
Moreover, if $X$ is a closed subspace of a locally convex Hausdorff space $V$ and $v \not \in X$, then there exists a continuous linear functional $\lambda:V \to \bbC$ such that $\lambda|_X \equiv 0$ and $\lambda(v)=1$ \cite[Thm. 7.7.7]{NariciBeckenstein11}.
In contrast, there exist topological vector spaces which do not admit nonzero continuous linear functionals, such as $L^p$ spaces with $0 < p < 1$.

Most familiar examples of topological vector spaces, such as normed vector spaces, are locally convex.
Another source of locally convex spaces is via weak topologies \cite[\S8.2]{NariciBeckenstein11}.
Given a vector space $V$ and a set of linear functionals $\cX$ on $V$, the \textbf{weak} topology (or initial topology) on $V$ corresponding to $\cX$ is the coarsest topology making all of the functionals continuous.
This is a locally convex vector topology, being the seminorm topology corresponding to the seminorms $\abs{\lambda}$ for $\lambda \in \cX$.
A sequence (or net) $v_n \in V$ converges to $v$ if and only if $\lambda(v_n) \to \lambda(v)$ for every $\lambda \in \cX$.
A map $T:X \to V$ is continuous with respect to the weak topology if and only if $\lambda \circ T$ is continuous for every $\lambda \in \cX$.

Dually, we have the notion of the \textbf{colimit} (or final or strong) topology.
Consider a vector space $V$, and a family of linear maps $T_s:X_s \to V$ from topological vector spaces $X_s$ such that the images $T_s(X_s)$ span $V$.
The colimit topology on $V$ corresponding to the maps $T_s$ is the finest topology on $V$ such that every $T_s$ is continuous, and it is a vector topology \cite[\S4.11]{NariciBeckenstein11}.
If $U$ is a topological vector space, then a linear map $T:V \to U$ is continuous if and only if $T \circ T_s$ is continuous for all $s$.

If each space $X_s$ is locally convex then we may define a subtly different \textbf{locally convex colimit} topology on $V$, which is the finest locally convex topology such that each $X_s$ is continuous \cite[\S12.2]{NariciBeckenstein11}.
If $U$ is a \emph{locally convex} space then a linear map $T: V \to U$ is continuous for the locally convex colimit topology if and only if $T \circ T_s$ is continuous for all $s$ \cite[Thm. 12.2.2]{NariciBeckenstein11}.

We now discuss tensor products of locally convex spaces.
If $U$,$V$, and $X$ are vector spaces then bilinear maps $U \times V \to X$ correspond to linear maps $U \otimes V \to X$, where $\otimes$ is the algebraic tensor product.
If $U$,$V$, and $X$ are locally convex spaces, then there is a unique locally convex topology on $U \otimes V$, called the \textbf{$\pi$-topology} (or \textbf{projective topology}), such that \emph{jointly continuous} bilinear maps $U \times V \to X$ correspond to \emph{continuous} linear maps $U \otimes V \to X$ \cite[Prop. 43.4]{Treves67}.
We write $U \otimes_\pi V$ for the algebraic tensor product equipped with the $\pi$ topology.

We now conclude by revisiting the $\cF$-strong topology.
Suppose that $\cF$ is a set of operator-valued distributions on $S^1$ with domain a vector space $\cD$.
For every $\varphi_1, \ldots, \varphi_k \in \cF$ and $\Phi \in \cD$ we have a multilinear map
$C^\infty(S^1)^k \to \cD$ given by $(f_1, \ldots, f_k) \mapsto \varphi_1(f_1) \cdots \varphi_k(f_k)\Phi$.
These correspond to linear maps 
\[
S_{\varphi_1, \ldots, \varphi_k,\Phi}:C^\infty(S^1) \otimes_\pi \cdots \otimes_\pi C^\infty(S^1) \to \cD.
\]
We include the case $k=0$, in which case $S_\Phi:\bbC \to \cD$ assigns $1 \mapsto \Phi$.
The $\cF$-strong topology on $\cD$ is then defined to be the locally convex colimit of the maps $S_{\varphi_1, \ldots, \varphi_k, \Phi}$.
Unpacking the definitions, if $X$ is a locally convex space then a map $T:\cD \to X$ is $\cF$-strong continuous if and only if $T(\varphi_1(f_1) \cdots \varphi_k(f_k)\Phi)$ is jointly continuous in the $f_j$ for all $\varphi_1, \ldots, \varphi_k \in \cF$ and $\Phi \in \cD$.


\def\lfhook#1{\setbox0=\hbox{#1}{\ooalign{\hidewidth \lower1.5ex\hbox{'}\hidewidth\crcr\unhbox0}}}

\end{document}